\documentclass[12pt]{article}

\RequirePackage[OT1]{fontenc}
\usepackage{amsthm,amsmath,amssymb,bm,enumerate,xcolor,graphicx,natbib}
\RequirePackage[colorlinks,citecolor=blue,urlcolor=blue]{hyperref}
\usepackage{xr}

\renewcommand{\Pr}{\mbox{P}}
\newcommand{\E}{{\rm E}}
\newcommand{\Exp}{\mbox{Exp}}
\newcommand{\etaW}{\eta_W}
\newcommand{\etaX}{\eta_X}
\newcommand{\etaWh}{\eta_W(h)}
\newcommand{\bs}{\bm{s}}

\newcommand{\blind}{1}

\newtheorem{rmk}{Remark}
\newtheorem{prop}{Proposition}
\newtheorem{cor}{Corollary}
\newtheorem{lem}{Lemma}

\theoremstyle{definition}
\newtheorem{example}{Example}

\def\spacingset#1{\renewcommand{\baselinestretch}%
{#1}\small\normalsize} \spacingset{1}

\begin{document}

\if1\blind
{
  \title{\bf Modeling spatial processes with unknown extremal dependence class}
  \author{Rapha\"{e}l G. Huser \\
    CEMSE Division, King Abdullah University of Science and Technology\\
    and \\
   Jennifer L. Wadsworth
   \\
    Department of Mathematics and Statistics, Lancaster University
    }
  \maketitle
} \fi

\if0\blind
{
  \bigskip
  \bigskip
  \bigskip
  \begin{center}
    {\LARGE\bf Modeling spatial processes with unknown extremal dependence class}
\end{center}
  \medskip
} \fi

\bigskip
\begin{abstract}
Many environmental processes exhibit weakening spatial dependence as events become more extreme. Well-known limiting models, such as max-stable or generalized Pareto processes, cannot capture this, which can lead to a preference for models that exhibit a property known as asymptotic independence. However, weakening dependence does not automatically imply asymptotic independence, and whether the process is truly asymptotically (in)dependent is usually far from clear. The distinction is key as it can have a large impact upon extrapolation, i.e., the estimated probabilities of events more extreme than those observed. In this work, we present a single spatial model that is able to capture both dependence classes in a parsimonious manner, and with a smooth transition between the two cases. The model covers a wide range of possibilities from asymptotic independence through to complete dependence, and permits weakening dependence of extremes even under asymptotic dependence. Censored likelihood-based inference for the implied copula is feasible in moderate dimensions due to closed-form margins. The model is applied to oceanographic datasets with ambiguous true limiting dependence structure.
\end{abstract}

\noindent%
\vfill

\newpage

\newpage

\baselineskip=26pt

\section{Introduction}
\label{sec:Introduction}
The statistical modeling of spatial extremes has received much attention since the article of \citet{Padoanetal10} provided a method of inference for \emph{max-stable processes}. The latter form an important class of models for spatial extremes, as they arise as the only non-degenerate limits of renormalized pointwise maxima of spatial stochastic processes. More precisely, let $Y_i(\bs)$, $i=1,2,\ldots$, be independent and identically distributed copies of a stochastic process  $\{Y(\bs):\bs\in\mathcal{S}\}$ with index set $\mathcal{S} \subset \mathbb{R}^2$. If there exist functions $a_n(\bs)>0, b_n(\bs)$ such that the limiting process
\begin{align}\label{MaxStable}
 M(\bs) = \lim_{n\to\infty} \max_{1\leq i \leq n}\frac{Y_i(\bs)-b_n(\bs)}{a_n(\bs)}
\end{align}
has non-degenerate marginals, then $M(\bs)$ is a max-stable process \citep[][Chapter 9]{deHaan.Ferreira:2006}. A practical issue with max-stable processes is that their $d$-dimensional densities (and hence the likelihood function) are difficult to evaluate, as the number of terms involved equals the $d$th Bell number, which grows super-exponentially with $d$. As such, spatial models for high threshold exceedances, which have simpler likelihoods, have become more appealing; see e.g., \citet{FerreiradeHaan14, WadsworthTawn14, Engelkeetal15, ThibaudOpitz15} and \citet{deFondevilleDavison16}. The threshold exceedance analogue of the max-stable process is known as the \emph{generalized Pareto process}, and has a similar asymptotic dependence structure in its joint tail region.

In order for limiting max-stable or generalized Pareto processes to provide good statistical models, we require that the extremes of $Y(\bm{s})$ are well represented by these processes, i.e., adequate convergence has occurred. However, there are no guarantees on rates of convergence, and in practice, limit models may not hold well. One way to assess the validity of convergence is to assess whether the stability properties of limit models hold well: max-stable copulas are invariant to the operation of taking pointwise maxima (max-stability), whilst generalized Pareto copulas are invariant to conditioning on threshold exceedances of higher levels (threshold-stability). A graphical diagnostic for max-stability is presented in \citet{Gabdaetal12}, whilst for threshold-stability, one can examine plots of
\begin{align}
 \chi_u := \Pr\{F_1(Y_1)>u,\ldots,F_d(Y_d)>u|F_1(Y_1)>u\},\qquad Y_j=Y(\bs_j) \sim F_j, \label{eq:chiu:intro}
\end{align}
where the argument $\bs_j$ denotes the $j$th spatial location; if the data follow a generalized Pareto process law, then this function should be constant as the quantile $u$ tends to one \citep{Rootzenetal16}. For environmental data in particular, it is much more common to see estimates of~\eqref{eq:chiu:intro} decreasing as $u\to 1$, indicating that dependence weakens with level of extremeness. An example of this is given in Figure~\ref{fig:waveheightchi}, for a dataset of significant wave heights, to be analyzed in \S\ref{sec:WaveHeight}.

\begin{figure}
 \centering
 \includegraphics[width=0.6\textwidth]{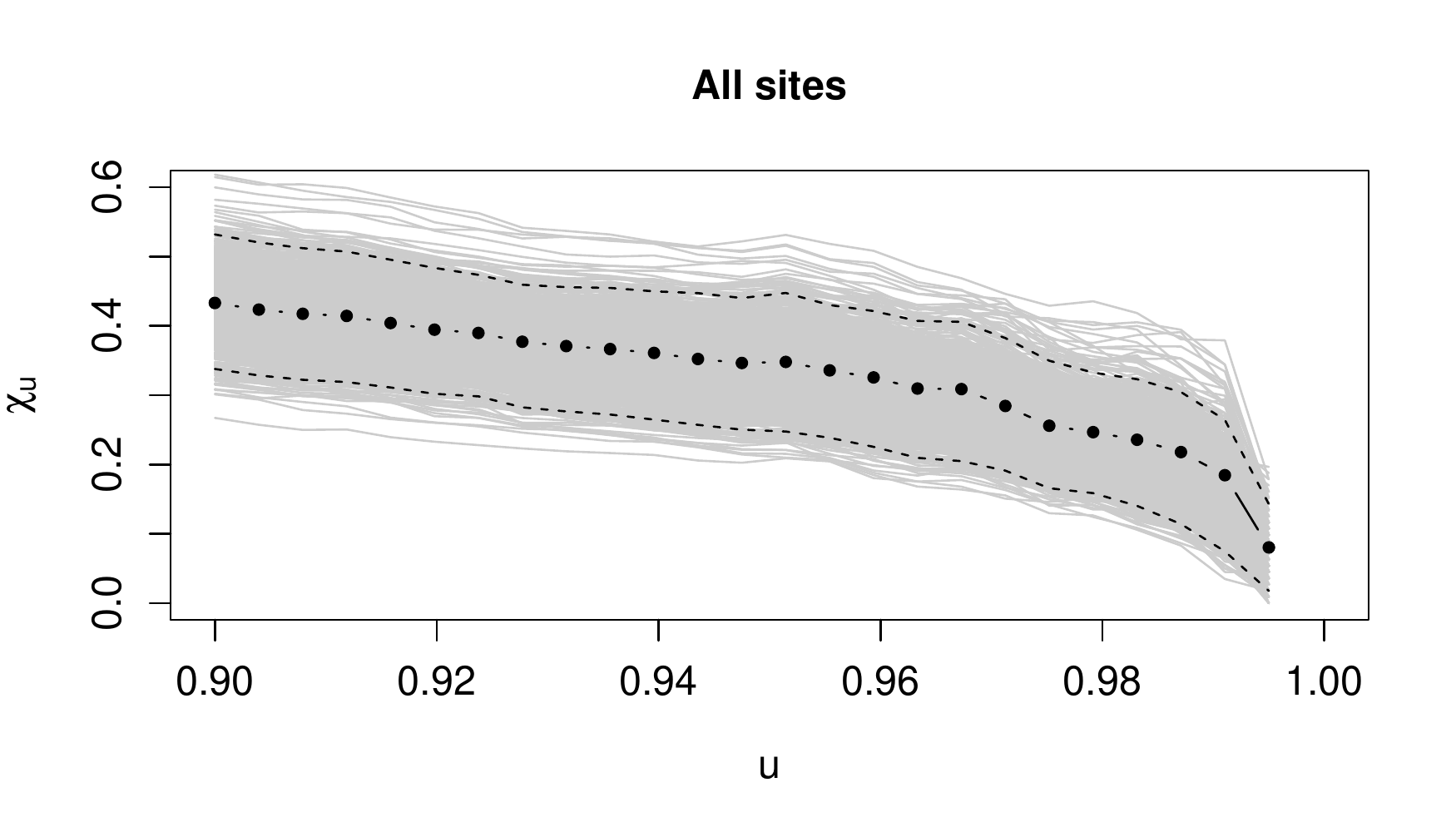}
 \label{fig:waveheightchi}
 \caption{Estimate of the dependence summary $\chi_u$ in \eqref{eq:chiu:intro} (dots) for the significant wave height data of \S\ref{sec:WaveHeight} plotted against quantile, $u$. Gray lines are estimates from a stationary bootstrap resampling procedure (see \S\ref{sec:WaveHeight}); dashed lines indicate the central 95\% of bootstrap samples, pointwise.}
\end{figure}

If the limit of $\chi_u$ defined in~\eqref{eq:chiu:intro} as $u\to1$ is positive for all sites $\bs_1,\ldots,\bs_d$ and all $d\geq2$, the process $Y(\bs)$ is termed \emph{asymptotically dependent}, and eventually, possibly at much higher levels, a generalized Pareto process should represent a suitable model for the data. If the limit is zero for all sites $\bs_1,\ldots,\bs_d$ and all $d\geq2$, we term the process \emph{asymptotically independent}; in such cases no generalized Pareto model would ever be suitable. Intermediate scenarios are possible, but owing to the structure of spatial data, it is common over small spatial domains to assume that the process is either asymptotically dependent or asymptotically independent, and we assume this here also. Determining a suitable model for the data usually requires distinguishing between these two scenarios, since most models exhibit only one type of dependence; choosing the incorrect class  will lead to unsuitable extrapolation into the joint upper tail \citep{LedfordTawn97, Davisonetal13}.

In practice, because asymptotic properties are always difficult to infer, it is ideal to fit spatial models encompassing both asymptotic dependence classes, and let the data speak for themselves. To our knowledge, the only instance in the literature of such a hybrid spatial extreme model is the max-mixture model of \citet{WadsworthTawn12}. However, in that model, asymptotic independence only occurs at a boundary point of the parameter space, thus inference methods allowing for this are non-regular. Moreover, the model is highly parametrized and requires pairwise likelihood fitting methods.
 
 In this paper, we address such deficiencies by presenting a class of spatial processes described by a small number of parameters and making a smooth transition between the two dependence paradigms. Specifically, we propose a novel class of spatial extremal models that have non-trivial asymptotically dependent and asymptotically independent submodels with the transition taking place in the interior of the parameter space. The latter property allows us to quantify our uncertainty about the dependence class in a simple manner. Our new spatial models can thus be viewed as similar in spirit to the generalized extreme-value (GEV) distribution in the univariate case, which was introduced by \citet{VonMises:1954} and \citet{Jenkinson55} as a three-parameter model combining the three limiting extreme-value types (i.e., reversed Weibull, Gumbel and Fr\'echet), hence providing a way to make inference without specifying the asymptotic distribution family prior to fitting the model. Furthermore, subject to model assumptions, standard hypothesis testing methods can be used to assess the evidence for asymptotic dependence over asymptotic independence, if so desired.

In encompassing both extremal dependence classes, our approach has similarities with the bivariate model of \citet{Wadsworthetal17}. However our construction here is simpler and substantially more amenable to higher-dimensional inference. Other related work that allows for both asymptotic dependence structures in a spatial setting is the Gaussian scale mixture models proposed in the recent work of \citet{Huseretal17}, but their models either make the transition at a boundary point of the parameter space, or are inflexible in their representation of asymptotic independence structures.

The paper is organized as follows. Section~\ref{sec:Model} describes the new spatial model and its extremal dependence properties. Section~\ref{sec:InferenceSimulation} details censored likelihood inference, describes a test for the asymptotic dependence class, and presents a simulation study validating the methodology. The new model is then applied to two oceanographic datasets in Section~\ref{sec:Oceanographic}, while Section~\ref{sec:Discussion} concludes with some discussion. All proofs are deferred to Appendix~\ref{app:proofs}.

\section{Model}
\label{sec:Model}
\subsection{Copula-based approach}
\label{sec:ModelIntroduction}

 The main goal of this work is to provide flexible extremal dependence structures for spatial processes. As such, we take a copula-based approach and seek the construction of flexible families of copulas for spatial extremal dependence. For a process with marginal distribution functions $X_j \sim F_j$, the $d$-dimensional copula function $C$, is defined as
\begin{align*}
 C(u_1,\ldots,u_d) &= \Pr\{F_1(X_1) \leq u_1, \ldots, F_d(X_d) \leq u_d\}.
\end{align*}
When the margins $F_j$, $j=1,\ldots,d$, are continuous, which will be the case throughout this paper, the copula is unique \citep{Sklar59}, and represents a multivariate distribution function with standard uniform margins. In \S\ref{sec:Construction}, we describe construction of a model whose copula displays interesting extremal dependence properties. Details of likelihood calculations of the copula of the model we introduce are presented in Section~\ref{sec:InferenceSimulation}.

\subsection{Construction}
\label{sec:Construction}
Let $\{W(\bs):\bs\in\mathcal{S}\subset \mathbb{R}^2\}$ be a stationary spatial process with standard Pareto margins, and displaying asymptotic independence with \emph{hidden regular variation}; a consequence of this is that for any $x\geq 1$,
\begin{align}
 \Pr\{W(\bs_j)>x\} &= x^{-1}, \notag\\
 \Pr\{W(\bs_j)>x, W(\bs_k)>x\} &= L_W(x) x^{-1/\eta_W(\bm{h})},~~~k\neq j, \label{eq:Whrv}
\end{align}
where $L_W:(0,\infty)\to(0,\infty)$ is slowly varying at infinity, i.e., ${L_W(ax)/L_W(x)\to1}$ as $x\to\infty$ for any $a>0$, and $0<\eta_W(\bm{h})<1$ for $\bm{h}=\bs_j-\bs_k\neq\bm{0}$ \citep{LedfordTawn96,Resnick02}. Note that we exclude the further possibility $\eta_W(\bm{h})=1$ ($\bm{h}\neq\bm{0}$), $L_W(x)\to 0$ as $x\to\infty$, because this does not arise in models that we might naturally consider for $W(\bs)$. The parameter $\eta_W(\bm{h})$, called the \emph{coefficient of tail dependence}, summarizes the joint tail decay of the process $W(\bs)$ and it is a function of the lag vector $\bm{h}$. For simplicity, in what follows we will restrict ourselves to isotropic processes, and will therefore write $\etaWh$ (or, for notational convenience, $\etaW$, when no confusion can arise), where $h=\|\bm{h}\|=\|\bs_j-\bs_k\|\geq0$ denotes the Euclidean distance between sites $\bs_j,\bs_k\in\mathcal{S}$. Examples of models satisfying \eqref{eq:Whrv} include marginally transformed Gaussian processes and inverted max-stable processes; see \S\ref{sec:ExampleModels} for more details.

With $W(\bs)$ as described, let $R$ be an independent standard Pareto random variable. Our spatial dependence model is defined through the random field constructed as
\begin{align}
X(\bs)=R^{\delta}W(\bs)^{1-\delta},~~~\delta\in[0,1]. \label{eq:model}
\end{align}
The following simple observation highlights why the parsimonious model defined in \eqref{eq:model} is potentially useful: when $\delta>1/2$ then $R^{\delta}$ is heavier-tailed than $W^{1-\delta}$ and this induces asymptotic dependence; when $\delta<1/2$ the converse is true, and this induces asymptotic independence. These facts are formalized in \S\ref{sec:Dependence}.

Construction~\eqref{eq:model} has superficial similarities with the Gaussian scale mixture models studied by \citet{Huseretal17}, who multiply a Gaussian random field by a random effect that determines the extremal dependence properties. However, in~\eqref{eq:model} the latent process $W(\bs)$ does not have Gaussian margins, resulting in a very different construction in practice, and need not have a Gaussian copula structure, which yields a much wider class of models. In practice, high-dimensional inference requires tractable densities for $W(\bs)$ (see \S\ref{sec:Likelihood}), leading to the Gaussian copula as a natural choice in spatial settings. Alternative possibilities for $W(\bs)$ are discussed further in \S\ref{sec:ExampleModels}.

\begin{rmk}
 Representation \eqref{eq:model} is convenient to study the asymptotic dependence properties of the process $X(\bs)$ using the theory of regular variation; see \S\ref{sec:Dependence}. However, as the copula structure is invariant with respect to monotone marginal transformations, there is an infinite number of ways to characterize the copula stemming from $X(\bs)$, some of which may be computationally more attractive or have appealing interpretations. For example, taking the logarithm on both sides of \eqref{eq:model}, we obtain an additive structure
 \begin{align}
  \tilde{X}(\bs) &:= \delta \tilde{R} + (1-\delta)\tilde{W}(\bs),\label{eq:model.additive}
 \end{align}
where $\tilde{R}:=\log(R) \sim \Exp(1)$ is independent of $\tilde{W}(\bs):=\log\{W(\bs)\}$, also with $\Exp(1)$ margins. In Sections~\ref{sec:InferenceSimulation} and \ref{sec:Oceanographic}, copula and likelihood computations are based on expression~\eqref{eq:model.additive}.
\end{rmk}

The variable $R$ in \eqref{eq:model} or equivalently the variable $\tilde{R}$ in \eqref{eq:model.additive}, may be interpreted in various ways, shedding light on the extremal behavior of $X(\bs)$. For example, by writing $\tilde{R}:=\{\tilde{R}(\bs):\bs\in\mathcal S\subset\mathbb R^2\}$, it can be seen as a random process indexed by $\mathcal S$ with perfect dependence, so the representation in \eqref{eq:model.additive} implies that $\tilde{X}(\bs)$ can be interpreted as a mixture between perfect dependence and asymptotic independence. This contrasts with \citet{ColesPauli02}, who constructed hybrid bivariate models using a certain type of mixture between asymptotic dependence and complete independence. 

Alternatively, $R$ or $\tilde{R}$ may be interpreted as an unobserved latent random factor impacting simultaneously the whole region $\mathcal S$, hence affecting the joint tail characteristics, and making a link with the common factor copula models for spatial data introduced by \citet{Krupskii.etal:2017}. One major difference with our approach, however, is that here $\tilde{R}$ and $\tilde{W}(\bs)$ are both on the unit exponential scale, whereas the location mixture copula models of \citet{Krupskii.etal:2017} assume that $\tilde{W}(\bs)$ is a Gaussian process and that both components in \eqref{eq:model.additive} are weighted equally, corresponding to $\delta=1/2$. Consequently, their exponential factor model always displays asymptotic dependence. Other distributions for the random factor were investigated in \citet{Krupskii.etal:2017}, but they all yield copulas with (non-trivial) asymptotic dependence lying on the boundary of, or at a single point in, the parameter space.

We next study the dependence properties of model~\eqref{eq:model} for $\delta\in(0,1)$ noting the simple interpretations at the endpoints of the parameter space: it is clear from~\eqref{eq:model} or~\eqref{eq:model.additive} that perfect dependence arises as $\delta\to1$, whilst the copula of $W(\bs)$ is recovered as $\delta\to0$. 

\subsection{Dependence properties}
\label{sec:Dependence}
Owing to the simple construction of this process, it is sufficient to study bivariate dependence to make more general conclusions. Comments on higher-dimensional dependence will be made throughout the remainder of the section.

To examine the dependence properties of the process~\eqref{eq:model}, we relate the behavior of the bivariate joint survivor function on the diagonal, $\Pr(X_j>x, X_k>x)$, to the marginal survivor function, $\Pr(X_j>x)$, where for simplicity we write $X_j=X(\bs_j)$ and so forth. We focus on a bivariate version of the dependence measure~\eqref{eq:chiu:intro},
\begin{align}\label{eq:chidef}
\chi_u(h):=\Pr\{F_j(X_j)>u \mid F_k(X_k)>u\},\qquad\mbox{and}\qquad X_j\sim F_j, X_k\sim F_k,
\end{align}
and its limit $\chi(h):= \lim_{u\to 1}\chi_u(h)$, with $h=\|\bs_j-\bs_k\|$. A value of $\chi(h)>0$ indicates asymptotic dependence for this pair of sites, whilst $\chi(h)=0$ defines asymptotic independence. Because the process $X(\bs)$ has common margins with upper endpoint at infinity, the limit may be equivalently expressed as
\begin{align}
 \chi(h)= \lim_{x\to \infty} \Pr(X_j>x, X_k>x) / \Pr(X_j>x). \label{eq:chih}
\end{align}
 When $\chi(h)=0$, alternative measures are needed to discriminate between the different levels of dependence exhibited by asymptotically independent distributions. A widely satisfied assumption, already made for the process $W(\bs)$ in \eqref{eq:Whrv} (modulo the restriction made on the coefficient of tail dependence), is
\begin{align}
 \Pr(X_j>x, X_k>x) = L_X\{\Pr(X_j>x)^{-1}\} \Pr(X_j>x)^{1/\etaX(h)}, \label{eq:etaX}
\end{align}
where $L_X$ is slowly varying at infinity, and $\etaX(h)\in(0,1]$ is the coefficient of tail dependence for the process $X(\bs)$. When $\etaX(h)=1$ and $L_X(x)\to \chi(h)>0$ as $x\to\infty$, the pair of variables $(X_j,X_k)^T$ are asymptotically dependent, else they are asymptotically independent and the value of $\etaX(h)$ summarizes the strength of extremal dependence in the joint upper tail. For notational convenience, the dependence on distance $h$ in $\chi(h)$ and $\eta_X(h)$ may be omitted when no confusion can arise.

\subsubsection{Marginal distribution}
\label{sec:Marginal}
The marginal distribution of the process~\eqref{eq:model} may be established for $\delta\neq1/2$ as follows:
\begin{align}
 1-F_{X}(x)=\Pr(X_j>x) &=\Pr\{W_j > x^{1/(1-\delta)}R^{-\delta/(1-\delta)}\}\notag\\
 &= \Pr(R>x^{1/\delta}) + x^{-1/(1-\delta)} \int_1^{x^{1/\delta}} r^{\delta/(1-\delta)-2} \,{\rm d}r \notag\\
 &= \frac{\delta}{2\delta-1} x^{-1/\delta} - \frac{1-\delta}{2\delta-1} x^{-1/(1-\delta)},\qquad x\geq 1, \delta\neq1/2. \label{eq:marg}
\end{align}
The case $\delta=1/2$ may either be established independently, or as a limit, from which we get
\begin{align*}
  1-F_{X}(x)=\Pr(X_j>x) = x^{-2}\{2\log(x) + 1\}\qquad x\geq 1, \delta=1/2;
\end{align*}
this is the survival function of a log-Gamma random variable with rate and shape parameters both equal to two. Notice that margins are here available in closed form, unlike the Gaussian scale mixture model of \citet{Huseretal17}, or the bivariate model of \citet{Wadsworthetal17}. Since the copula is the object of interest in all of the above cases, this makes model~\eqref{eq:model} computationally more appealing.

\subsubsection{Joint distribution}
\label{sec:Joint}
We now derive the joint survivor function of a pair of variables $(X_j,X_k)^T$ from the process $X(\bs)$ in \eqref{eq:model}, and then use this result in \eqref{eq:chih} and~\eqref{eq:etaX}, combined with \eqref{eq:marg}, to derive the corresponding coefficients $\chi$ and $\eta_X$ characterizing tail dependence of $X(\bs)$ depending on the value of $\delta$.
\begin{prop}
\label{prop:joint}
 With definitions and notation as established above, the joint survivor function of~\eqref{eq:model} satisfies
 \begin{align*}
  \Pr(X_j>x, X_k>x) &= \begin{cases}
                        L(x)x^{-1/\{(1-\delta)\etaW\}}, & \mbox{if } \etaW \geq \delta/(1-\delta),\\
                        \E\{\min(W_j,W_k)^{(1-\delta)/\delta}\}x^{-1/\delta}\{1+o(1)\} & \mbox{if } \etaW < \delta/(1-\delta),
                       \end{cases}
 \end{align*}
 where $L$ is slowly varying at infinity.
\end{prop}

\begin{cor}
\label{cor:joint}
If $\delta > 1/2$, the pair $(X_j,X_k)^T$ is asymptotically dependent with
\begin{align}
 \chi &= \E\left\{\min(W_j,W_k)^{(1-\delta)/\delta}\right\}\frac{2\delta-1}{\delta} =\E\left(\min\left[\frac{W_j^{(1-\delta)/\delta}}{\E\{W_j^{(1-\delta)/\delta}\}}, \frac{W_k^{(1-\delta)/\delta}}{\E\{W_k^{(1-\delta)/\delta}\}}\right]\right)>0. \label{eq:chiW}
\end{align}
If $\delta \leq 1/2$, the pair $(X_j,X_k)^T$ is asymptotically independent, i.e., $\chi=0$. Furthermore, the coefficient of tail dependence for the process~\eqref{eq:model} is
\begin{align}
 \etaX &=\begin{cases}
	    1, & \mbox{if } \delta \geq 1/2,\\
           \delta/(1-\delta), & \mbox{if } \etaW / (1+\etaW) < \delta < 1/2,\\
           \etaW, & \mbox{if } \delta \leq \etaW / (1+\etaW).
          \end{cases}\label{eq:etaX2}
\end{align}
\end{cor}

\begin{rmk}
 Analogous dependence summaries in $d$ dimensions are simple to establish using the same techniques of proof as for Proposition~\ref{prop:joint} and Corollary~\ref{cor:joint}. Specifically, letting $\etaX^{1:d}$ and $\etaW^{1:d}$ denote $d$-dimensional counterparts of the coefficient of tail dependence, defined using the $d$-dimensional joint survivor function, then expression~\eqref{eq:etaX2} still holds with $\etaX$ and $\etaW$ replaced by $\etaX^{1:d}$ and $\etaW^{1:d}$. The $d$-dimensional analogue of $\chi$ generalizes expression~\eqref{eq:chiW}, and is discussed in Remark~\ref{rmk:thetachi}.
\end{rmk}

The case $\delta=1/2$ is of particular interest, since it represents a boundary between asymptotic dependence and asymptotic independence: according to Corollary~\ref{cor:joint} we have asymptotic independence ($\chi=0$), but the coefficient of tail dependence $\etaX$ attains its boundary value of 1. In this case, we therefore have $L_X(x)\to0$ as $x\to\infty$ in \eqref{eq:etaX}. Furthermore, the model has the appealing property that $\chi \searrow 0$ as $\delta \searrow 1/2$ and $\etaX \searrow \etaW$ as $\delta\searrow \etaW/(1+\etaW)$. As noted at the end of \S\ref{sec:Construction}, as $\delta\searrow0$, the dependence structure of the $W$ process is recovered. Our model $X(\bs)$ in \eqref{eq:model} hence provides a smooth interpolation from the asymptotically independent submodel $W(\bs)$ and perfect dependence, as the parameter $\delta$ varies in the unit interval, and it transits through non-trivial asymptotically independent and asymptotically dependent submodels.

\subsection{Further dependence properties under asymptotic dependence}
\label{sec:FurtherPropertiesAD}

Here, we outline the connection to other well-known measures of dependence in the case of asymptotic dependence. We focus firstly on a limiting measure, namely the so-called \emph{exponent function}, $V:(0,\infty)^d \to (0,\infty)$ defined for all $x_1,\ldots,x_d>0$ by
\begin{align*}
 V(x_1,\ldots,x_d) = \lim_{t\to\infty} t\left(1-\Pr\left[X_1 \leq F_{X}^{-1}\left\{1-(tx_1)^{-1}\right\},\ldots,X_d \leq F_{X}^{-1}\left\{1-(tx_d)^{-1}\right\}\right]\right),
\end{align*}
which describes the joint dependence of the associated max-stable or generalized Pareto process; see \citet{Davison.etal:2012}, \citet{Cooley.etal:2012a}, \citet{Segers12} or \citet{Davison.Huser:2015} for recent reviews on max-stable models. We then examine the \emph{sub-asymptotic} behavior under asymptotic dependence, i.e., the mode of convergence towards such limiting structures, which is important in practice for modeling extreme events at observable levels.

\begin{prop}
\label{prop:chiV}
  For $X(\bs)$ as given in model~\eqref{eq:model} and $\delta>1/2$,
\begin{align*}
 V(x_1,\ldots,x_d) &=\E\left[\max_{j=1,\ldots,d} \frac{W_j^{(1-\delta)/\delta}}{\E\{W_j^{(1-\delta)/\delta}\}x_j}\right]=\frac{2\delta-1}{\delta}\E\left\{\max_{j=1,\ldots,d} \frac{W_j^{(1-\delta)/\delta}}{x_j}\right\}.
\end{align*}
\end{prop}

\begin{rmk}
\label{rmk:thetachi}
 It follows from Proposition~\ref{prop:chiV} that the $d$-dimensional \emph{extremal coefficient}, $\theta^{1:d} = V(1,\ldots,1)$, is
\begin{align}
\theta^{1:d} &=\E\left[\max_{j=1,\ldots,d} \frac{W_j^{(1-\delta)/\delta}}{\E\{W_j^{(1-\delta)/\delta}\}}\right]=\frac{2\delta-1}{\delta}\E\left\{\max_{j=1,\ldots,d} W_j^{(1-\delta)/\delta}\right\}\in[1,d]. \label{eq:theta}
\end{align}
 Furthermore the $d$-dimensional version of $\chi$, $\chi^{1:d}$, in \eqref{eq:chidef} equals
\begin{align*}
\lim_{u\to 1} \Pr\{F_X(X_1)>u,\ldots,F_X(X_d)>u \mid F_X(X_1)>u\}=\frac{2\delta-1}{\delta}\E\left\{\min_{j=1,\ldots,d} W_j^{(1-\delta)/\delta}\right\},
\end{align*}
which can be demonstrated directly or by using inclusion-exclusion relationships between these parameters, such as outlined in \citet{Rootzenetal16}.
\end{rmk}

The behavior of $\chi_u-\chi = \Pr\{F_X(X_j)>u \mid F_X(X_k)>u\} - \chi$ as $u\to1$, i.e., the rate at which $\chi_u$ converges to its limit $\chi$, determines the flexibility of a process for capturing sub-asymptotic extremal dependence. Proposition~\ref{prop:chiconv} demonstrates that the parameterization of model~\eqref{eq:model} gives flexibility in this rate, meaning that dependence can weaken above the level used for fitting, whilst still allowing for the possibility of asymptotic dependence. 

\begin{prop}
\label{prop:chiconv}
 For $\delta>1/2$,
 \begin{align}
  \chi_u-\chi =   \chi \frac{1-\delta}{\delta}\left(\frac{\delta}{2\delta-1}\right)^{(1-2\delta)/(1-\delta)}(1-u)^{(2\delta-1)/(1-\delta)}\{1+o(1)\},\qquad u\to 1. \label{eq:chiconv}
 \end{align}
\end{prop}
For comparison, generalized Pareto processes have $\chi_u \equiv \chi$ for all $u$ above a certain level \citep{Rootzenetal16}, whilst all max-stable processes have $\chi_u-\chi \asymp (1-u)$, as $u\to1$. However, as $\chi$ is a dependence measure on the scale of the observations rather than maxima, it is less useful in the context of max-stable processes, where the summary~\eqref{eq:theta} is typically used instead. From Proposition~\ref{prop:chiconv}, we observe a wide range of convergence rates, from very rapid for $\delta$ near $1$, to rates slower than $(1-u)$ for $1/2<\delta<2/3$. Note that for $\delta < 1/2$, the rate $\chi_u-\chi$ is determined by the coefficient of tail dependence, $\etaX$; recall \eqref{eq:etaX} and \eqref{eq:etaX2}.

Figure~\ref{fig:chi.eta} illustrates the flexibility in extremal dependence structures, by plotting $\chi_u$ in \eqref{eq:chidef} as a function of the threshold $u$ and the limit quantity $\chi=\lim_{u\to1}\chi_u$ in \eqref{eq:chiW}, for a range of values of $\delta\in(0,1)$, and $(W_j,W_k)^T$ following a Gaussian copula with correlation parameter $0.4$. Figure~\ref{fig:chi.eta} also displays the coefficient of tail dependence $\eta_X$ defined in \eqref{eq:etaX} and \eqref{eq:etaX2} as a function of $\delta$ for $\eta_W=0.1,\ldots,0.9$. The smooth transition from asymptotic dependence to asymptotic independence taking place at $\delta=1/2$ can be clearly seen from these two plots. Moreover, as is intuitive, the right panel of Figure~\ref{fig:chi.eta} shows that the process $X(\bs)$ in \eqref{eq:model} cannot reach lower levels of dependence than its underlying $W(\bs)$ process. 

\begin{figure}[t!]
\centering
\includegraphics[width=\linewidth]{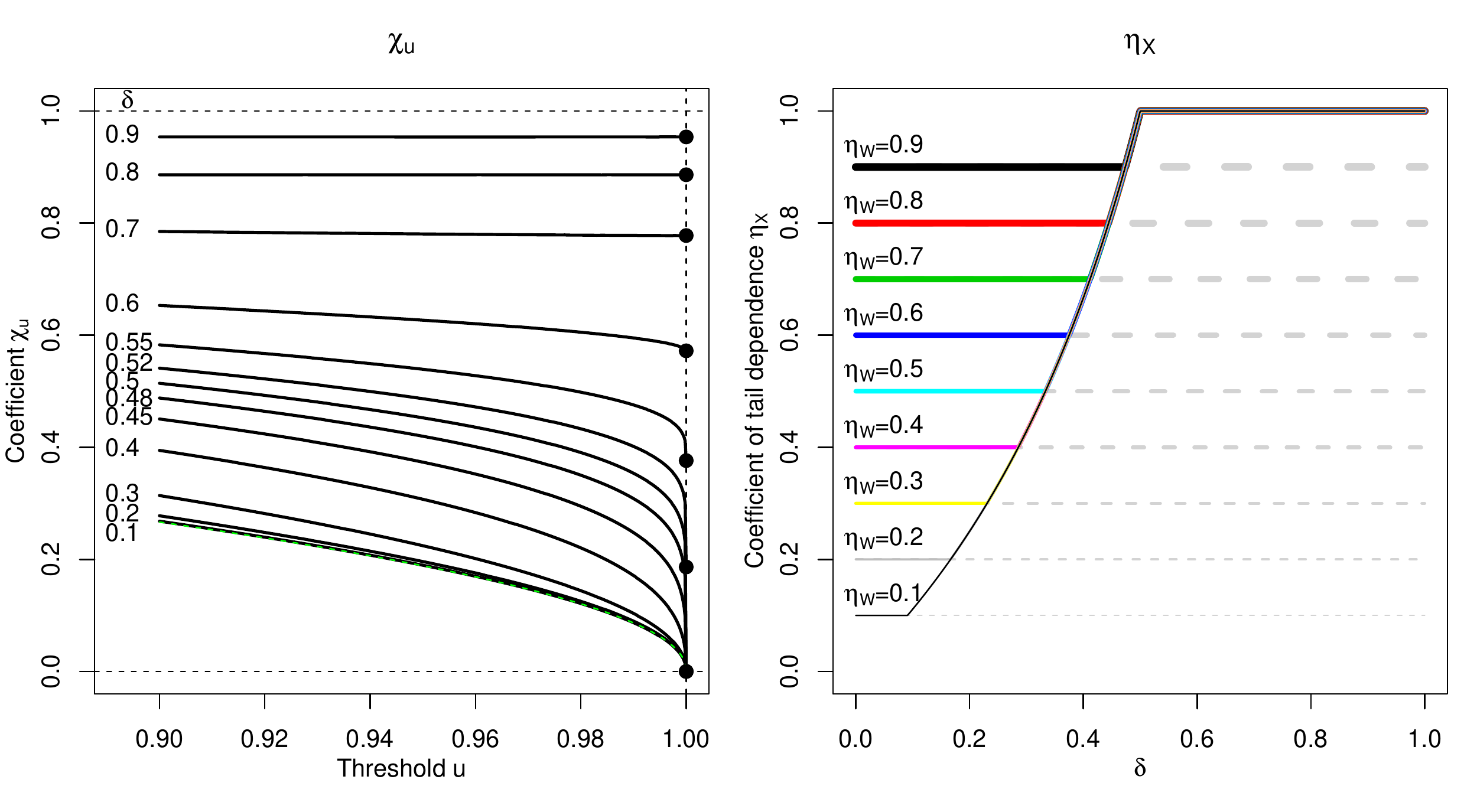}
\caption{\emph{Left:} Coefficient $\chi_u$ in \eqref{eq:chidef} for model \eqref{eq:model} plotted as a function of the threshold $u\in[0.9,1]$ for $\delta=0.1,0.2,0.3,0.4,0.45,0.48,0.5,0.52,0.55,0.6,0.7,0.8,0.9$ (top to bottom) and $(W_j,W_k)^T$ with a Gaussian copula with correlation $0.4$. Black dots at $u=1$ denote the limit quantities $\chi=\lim_{u\to1}\chi_u$, and the green dashed curve (superimposed with ${\delta=0.1}$) corresponds to the vector $(W_j,W_k)^T$, i.e., the case $\delta=0$. \emph{Right:} Coefficient of tail dependence $\eta_X$ in \eqref{eq:etaX2} (solid colored curves) as a function of $\delta\in(0,1)$ for $\eta_W=0.1,\ldots,0.9$ (thin to thick). The value $\eta_W=0.7$ corresponds to the choice of $(W_j,W_k)^T$ on the left panel.}\label{fig:chi.eta}
\end{figure}

\subsection{Example models}
\label{sec:ExampleModels}

We conclude this section with some concrete suggestions for the $W$ process that may be useful in certain applications, such as those described in Section~\ref{sec:Oceanographic}.

\begin{example}[Gaussian process]
\label{ex:WGaussian}
Let $\{Z(\bs):\bs\in\mathcal{S}\}$ be a stationary Gaussian process with correlation function $\rho(h)$, and standard Gaussian marginal distribution, denoted $\Phi$. Then
\[W(\bs) = 1/[1-\Phi\{Z(\bs)\}]\]
has a Gaussian copula, Pareto margins, and coefficient of tail dependence $\eta_W(h)=\{1+\rho(h)\}/2$. In this case, the value of $\chi(h)$ in~\eqref{eq:chiW} needs to be calculated either by Monte Carlo or numerical integration, both of which are simple and quick.
\end{example}

\begin{example}[Inverted max-stable process]
\label{ex:WinvertedMS}
Let $\{M(\bs):\bs\in\mathcal{S}\}$ be a stationary max-stable process with extremal coefficient function $\theta(h)\in(1,2]$, and marginal distribution functions $G_{\bs}$, $\bs\in\mathcal{S}$. Then the process
\[W(\bs) = 1/G_{\bs}\{M(\bs)\}\]
has an inverted max-stable copula \citep{LedfordTawn96,WadsworthTawn12}, Pareto margins, and coefficient of tail dependence $\eta_W(h)=1/\theta(h)$. The value of $\chi(h)$ in~\eqref{eq:chiW} can be calculated as
\begin{align}
 \chi(h) = \frac{2\delta-1}{\delta}\E\{\min(W_j,W_k)^{(1-\delta)/\delta}\} = \frac{2\delta-1}{1-(1-\delta)\{1+\etaW(h)\}},\quad\delta>1/2. \label{eq:chiIEV}
\end{align}
For this class of processes $\etaW(h) \in [1/2,1)$ for $h>0$, thus the range of $\chi(h)$ values can be established for each $\delta$. Moreover, the $d$-dimensional quantity $\chi^{1:d}$ takes the same form as in equation~\eqref{eq:chiIEV} with $\etaW(h)$ replaced by $\etaW^{1:d} \in[1/d,1)$.
\end{example}

In what follows, we will principally take $W(\bs)$ to have a Gaussian copula because the resulting density is much simpler in high dimensions than that of the inverted max-stable process, which suffers the same explosion in the number of terms as a max-stable process density. Pairwise or higher-dimensional composite likelihoods \citep[see, e.g.,][]{Padoanetal10,Varin.etal:2011,Castruccio.etal:2016} offer an alternative approach, but we do not explore this further here. Outside of a spatial context however, other dependence structures may be preferred. 

\begin{example}[Non-spatial model]
\label{sec:NonSpatial}
We remark that non-spatial use of the model~\eqref{eq:model} is also possible, replacing the process $W(\bs)$ with an asymptotically independent random vector $\bm{W}=(W_1,\ldots,W_d)^T$ with pairwise coefficients of tail dependence $\eta_{W}^{jk}<1$, $j<k\in\{1,\ldots,d\}$. For multivariate models in dimension $d$ greater than two some care is required, however, as model~\eqref{eq:model} allows only for $d$-wise asymptotic dependence (i.e., $\chi^{1:d}>0$), or $d$-wise asymptotic independence (i.e., $\chi^{jk}=0$, for all $j<k\in\{1,\ldots,d\}$). Such assumptions are natural in the context of spatial processes, but often less so for genuinely multivariate data. For dimension $d=2$ however, where $\chi^{1:2}>0$ is the complement of $\chi^{1:2}=0$, model~\eqref{eq:model} offers an interesting alternative to that of~\citet{Wadsworthetal17} for bivariate data. The latter show that the copula model defined by $\bm{X}=R\bm{W}$, where the radial variable $R$ follows a unit scale generalized Pareto distribution with shape parameter $\xi\in\mathbb R$ and $\max(\bm{W})=1$, with $R$ and $\bm{W}$ independent, displays asymptotic dependence for $\xi>0$ and asymptotic independence for $\xi\leq0$. One advantage of model~\eqref{eq:model} is that a version with an asymmetric dependence structure is simpler to implement, by selecting an asymmetric bivariate distribution for the copula of $\bm{W}$. We illustrate the improvement this can offer in \S\ref{sec:WaveSurge}.
\end{example}

\section{Inference and Simulation}
\label{sec:InferenceSimulation}

\subsection{Censored likelihood}
\label{sec:Likelihood}

We wish to fit the dependence structure of model~\eqref{eq:model} to the extremes of spatial processes. Since the dependence characteristics of the model are tailored towards appropriately capturing extremal dependence, we use a censored likelihood, which prevents low values from affecting the estimation of the extremal dependence structure. Such an approach is now standard in inference for multivariate and spatial extremes, although different censoring schemes have been adopted; see e.g.\ \citet{Smithetal97}, \citet{WadsworthTawn12} and \citet{Huseretal16,Huseretal17}. We assume that we are working with a $W$ process that has a density, so that this is also true for the copula.

Assume that $n$ independent replicates of a random process $\{Y(\bs):\bs\in\mathcal S\subset\mathbb R^2\}$ are observed at $d$ spatial locations, $\bs_1,\ldots,\bs_d\in\mathcal S$. Denote the $i$th replicate at the $j$th location by $Y_{ij}$, $i=1,\ldots,n$, $j=1,\ldots,d$. We assume that in its joint tail region, i.e., for observations above a high marginal threshold, the process $Y(\bs)$ has the same copula as our model $X(\bs)$ defined in \eqref{eq:model}, but with possibly different marginal distributions $F_{\bs}$. To estimate the dependence structure, we first transform the margins to uniform independently at each site $\bs_j$, $j=1,\ldots,d$. In Section~\ref{sec:Oceanographic}, we use the semi-parametric procedure of \citet{ColesTawn91}, whereby the distribution function is estimated using the asymptotically-motivated generalized Pareto distribution above a high marginal threshold, and the empirical distribution function below that threshold. The resulting variables are denoted $U_{ij} = \hat{F}_{\bs_j}(Y_{ij})$. An alternative is to use the empirical distribution function throughout as in \citet{Huseretal17}. This two-step approach is common practice in the copula literature and provides consistent inference for the copula under mild regularity conditions \citep[see, e.g.,][]{Joe:2015}.

The second step is to estimate the copula parameters using the transformed data based on a censored likelihood. When fitting the copula stemming from model~\eqref{eq:model}, the parameters to be estimated are $\boldsymbol{\psi}=(\delta,\boldsymbol{\psi}_W^T)^T\in\Psi={[0,1]\times\Psi_W}\subset\mathbb R^p$, where $\boldsymbol{\psi}_W$ is a $(p-1)$-dimensional vector of parameters describing the $W(\bs)$ process. Using the alternative representation $\tilde{X}(\bs)=\delta \tilde{R} + (1-\delta)\tilde{W}(\bs)$ in \eqref{eq:model.additive}, the resulting copula $C$ and its density $c$ are
\begin{align}
C(u_1,\ldots,u_d;\boldsymbol{\psi})&=F^{1:d}_{\tilde{X}}\{F_{\tilde{X}}^{-1}(u_1),\ldots,F_{\tilde{X}}^{-1}(u_d)\},\label{eq:CopulaCDF}\\
c(u_1,\ldots,u_d;\boldsymbol{\psi})&=f^{1:d}_{\tilde{X}}\{F_{\tilde{X}}^{-1}(u_1),\ldots,F_{\tilde{X}}^{-1}(u_d)\} \prod_{j=1}^d \left[f_{\tilde{X}}\{F_{\tilde{X}}^{-1}(u_j)\}\right]^{-1},\label{eq:CopulaPDF}
\end{align}
where $F_{\tilde{X}}(x)=F_{X}(e^x)$ and $f_{\tilde{X}}(x)=f_{X}(e^x)e^x$ with $f_{X}(x)={\rm d}F_{X}(x)/{\rm d}x$, easily obtained in closed form through \eqref{eq:marg}. The functions $F_{\tilde{X}}$ and $f_{\tilde{X}}$ are the marginal distribution and density, respectively, stemming from the $\tilde{X}(\bs)$ process observed at the sites $\bs_1,\ldots,\bs_d$, whilst
\begin{align*}
F^{1:d}_{\tilde{X}}(x_1,\ldots,x_d)&=\int_0^{r^\star_\delta}F_{\tilde{W}}\left\{(1-\delta)^{-1}(x_1-\delta r),\ldots,(1-\delta)^{-1}(x_d-\delta r)\right\}e^{-r}{\rm d}r,\\
f^{1:d}_{\tilde{X}}(x_1,\ldots,x_d)&=(1-\delta)^{-d}\int_0^{r^\star_\delta}f_{\tilde{W}}\left\{(1-\delta)^{-1}(x_1-\delta r),\ldots,(1-\delta)^{-1}(x_d-\delta r)\right\}e^{-r}{\rm d}r
\end{align*}
represent the joint distribution function and density, respectively, of this process. Here, $r^\star_\delta=\min(x_1,\ldots,x_d)/\delta$, and $F_{\tilde{W}},f_{\tilde{W}}$ denote the joint distribution and density, respectively, for the $\tilde{W}(\bs)$ process. The partial derivatives of the copula $C(u_1,\ldots,u_d;\boldsymbol{\psi})$ with respect to any set of variables $\mathcal J\subset\{1,\ldots,d\}$ of cardinality $d_{\mathcal J}$ may be expressed as
\begin{align}
C_{\mathcal J}(u_1,\ldots,u_d;\boldsymbol{\psi})&=\frac{\partial^{d_{\mathcal J}}}{\prod_{j\in\mathcal J}\partial u_j}C(u_1,\ldots,u_d;\boldsymbol{\psi})\notag\\
&=F^{1:d}_{\tilde{X},\mathcal J}\{F_{\tilde{X}}^{-1}(u_1),\ldots,F_{\tilde{X}}^{-1}(u_d)\} \prod_{j\in\mathcal J} \left[f_{\tilde{X}}\{F_{\tilde{X}}^{-1}(u_j)\}\right]^{-1},\label{eq:PartialCopula}
\end{align}
where
\begin{align*}
F^{1:d}_{\tilde{X},\mathcal J}(x_1,\ldots,x_d)&=(1-\delta)^{-d_{\mathcal J}}\int_0^{r^\star_\delta}F_{\tilde{W},\mathcal J}\left\{(1-\delta)^{-1}(x_1-\delta r),\ldots,(1-\delta)^{-1}(x_d-\delta r)\right\}e^{-r}{\rm d}r,
\end{align*}
with $F_{\tilde{W},\mathcal J}(x_1,\ldots,x_d)=\partial^{d_{\mathcal J}} F_{\tilde{W}}(x_1,\ldots,x_d)/\prod_{j\in\mathcal J}\partial x_j$. When the process $W(\bs)$ is based on a Gaussian copula, partial derivatives in \eqref{eq:PartialCopula} involve the multivariate Gaussian distribution in dimension $d-d_{\mathcal J}$. Although the unidimensional integrals appearing in \eqref{eq:CopulaCDF}, \eqref{eq:CopulaPDF} and \eqref{eq:PartialCopula} cannot be expressed in closed form, they can nevertheless be accurately approximated using standard finite integration or (quasi) Monte Carlo methods. To estimate the parameters $\boldsymbol{\psi}=(\delta,\boldsymbol{\psi}_W^T)^T\in\Psi$, while avoiding influence of non-extreme data below high marginal thresholds $u^\star_1,\ldots,u^\star_d$, we maximize the censored log likelihood function defined as
\begin{align}
\ell(\boldsymbol{\psi})&=\sum_{i=1}^n \log \{L_i(\boldsymbol{\psi})\},\label{eq:LogLikelihood}
\end{align}
with contributions defined through the sets of indices $\mathcal J_i=\{j:U_{ij}>u_j^\star\}\subseteq\{1,\ldots,d\}$ as
\begin{align*}
L_i(\boldsymbol{\psi})&=\begin{cases}
C(u_1^\star,\ldots,u_d^\star;\boldsymbol{\psi}), & \mathcal J_i=\emptyset,\\
c(U_{i1},\ldots,U_{id};\boldsymbol{\psi}), & \mathcal J_i=\{1,\ldots,d\},\\
C_{\mathcal J_i}\{\max(U_{i1},u_1^\star),\ldots,\max(U_{id},u_d^\star);\boldsymbol{\psi}\}, & \mbox{otherwise.}
\end{cases}
\end{align*}

The set $\mathcal J_i$ determines whether the $i$th observation vector $(U_{i1},\ldots,U_{id})^T$ has threshold exceedances in no, all, or some but not all components, respectively; therefore, these sets may be different for each likelihood contribution $i=1,\ldots,n$. The estimator maximizing \eqref{eq:LogLikelihood} over $\Psi$ is denoted by $\hat{\boldsymbol{\psi}}$. The performance of this inference approach is assessed in our simulation study \S\ref{sec:Simulation} and it is used in the application in \S\ref{sec:WaveHeight}.

Another possible censoring scheme is to use either the fully censored contribution $C(u_1^\star,\ldots,u_d^\star;\boldsymbol{\psi})$ in \eqref{eq:LogLikelihood} if $\mathcal J_i=\emptyset$ (i.e., the variable $U_{ij}$ is lower than the threshold $u_j^\star$ for all $j=1,\ldots,d$), or the completely uncensored contribution $c(U_{i1},\ldots,U_{id};\boldsymbol{\psi})$ otherwise. This was used by \citet{WadsworthTawn12}, \citet{Opitz16} and \citet{Wadsworthetal17}, and is adopted in the example of \S\ref{sec:WaveSurge}, where we compare fits of bivariate models.

\subsection{Simulation study}
\label{sec:Simulation}
\subsubsection{Parameter estimation}
\label{sec:ParameterEstimation}
To assess the performance of the maximum censored likelihood estimator $\hat{\boldsymbol{\psi}}$ defined through \eqref{eq:LogLikelihood}, we simulated data from the copula defined by \eqref{eq:model} at $d=2, 5,10,15$ locations uniformly generated in the unit square, $[0,1]^2$. We sampled $n=1000$ independent replicates at these locations, and considered the scenarios $\delta=0.1,\ldots,0.9$ (from asymptotic independence to dependence) with $W(\bs)$ defined by a Gaussian copula structure with powered exponential correlation function $\rho(\bs_j,\bs_k)=\exp\{-(\|\bs_j-\bs_k\|/\lambda)^\nu\}$, $\lambda>0,\nu\in(0,2]$. Setting $\lambda=0.5$ and $\nu=1$, we then estimated $\boldsymbol{\psi}=(\delta,\lambda,\nu)^T$ by maximizing \eqref{eq:LogLikelihood} with marginal thresholds $u_1^\star=\cdots=u_d^\star=0.95$, giving about $50$ exceedances at each location. For identifiability reasons, we fixed $\nu=1$ when $d=2$. Because the process is almost perfectly dependent when $\delta$ approaches unity, this creates numerical difficulties: to deal with this issue, we increase the relative precision in the {\tt R} function {\tt integrate} used for the calculation of the integrals in \eqref{eq:CopulaCDF}, \eqref{eq:CopulaPDF} and \eqref{eq:PartialCopula} for larger values of $\delta$. Figure~\ref{fig:results1} shows boxplots of estimated parameters $\hat\delta$ based on $100$ independent simulations. 

\begin{figure}[t!]
\centering
\includegraphics[width=\linewidth]{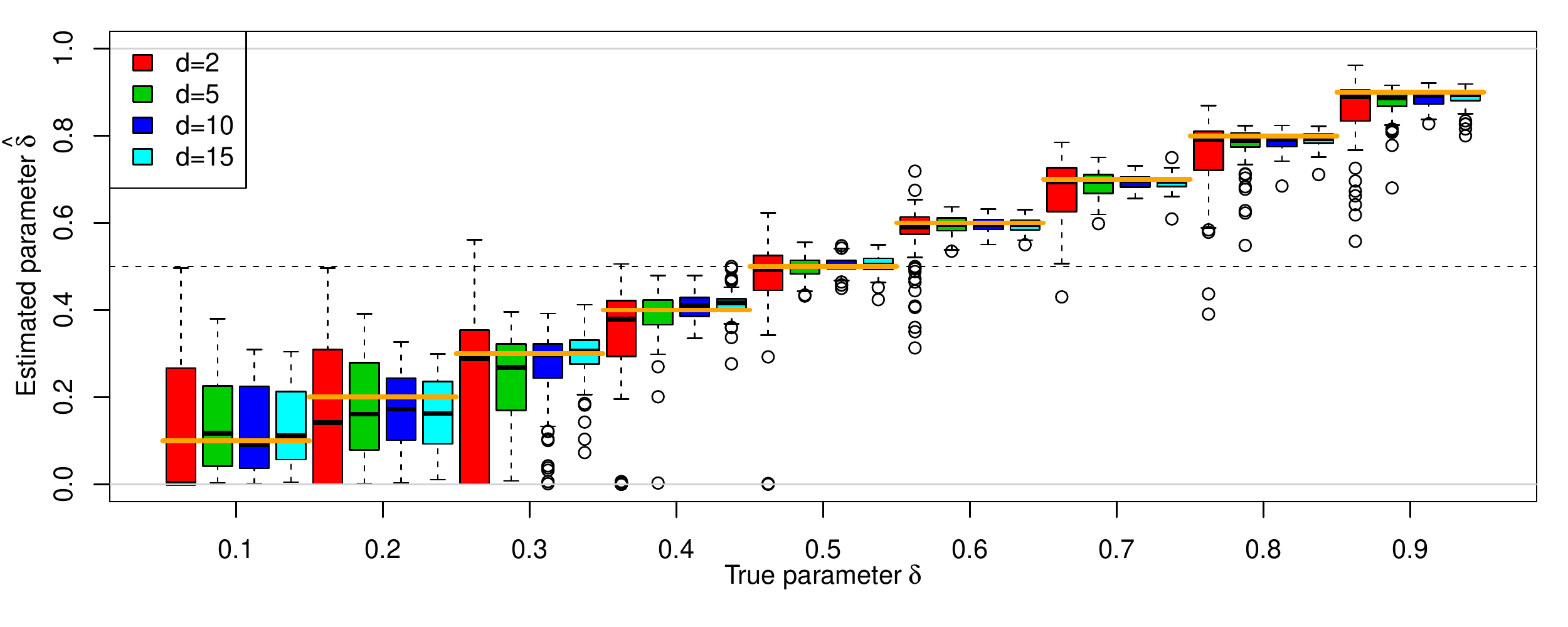}
\caption{Boxplots for the MLE $\hat\delta$ based on \eqref{eq:LogLikelihood} setting the thresholds as $u^\star_j=0.95$, $j=1,\ldots,d$, for model~\eqref{eq:model} with $\delta=0.1,\ldots,0.9$ (left to right). The underlying process $W(\bs)$ has a Gaussian copula with correlation function $\rho(\bs_j,\bs_k)=\exp\{-(\|\bs_j-\bs_k\|/\lambda)^\nu\}$ and $\lambda=0.5$, $\nu=1$. The process was observed at $d=2$ (red, left), $5$ (green, second left), $10$ (dark blue, second right) and $15$ (light blue, right) random locations in $[0,1]^2$, and $n=1000$ replicates were simulated. True values are indicated by the orange horizontal segments, and the boundary between asymptotic dependence and independence is indicated by the horizontal dashed line at $\delta=0.5$. Boxplots are based on $100$ independent simulations.}\label{fig:results1}
\end{figure}

Overall, the estimation procedure works as expected, with boxplots for $\delta$ approximately centered around the true value, though a small bias appears for $\delta=0.8,0.9$, which is due to numerical instabilities and difficulties in identifying all three parameters in such strong dependence scenarios, despite the higher numerical precision; recall Figure~\ref{fig:chi.eta}. As is typical for a bounded parameter, the asymptotic normality of $\hat\delta$ looks to hold well when $\delta$ is not too close to $0$ and $1$, but the distribution displays some asymmetry near to the endpoints $0$ and $1$. Estimation seems to be easier when $\delta\approx1/2$, which leads to small bias and variability. As $\delta\to0$, the copula structure of $X(\bs)$ converges to that of the latent process $W(\bs)$, here a Gaussian copula, and therefore low values such as $\delta=0.1,0.2$ yield very similar dependence structures, leading to higher variability.
Boxplots of $\hat\lambda$ and $\hat\nu$ (see Supplementary Material) suggest that results are better in the asymptotic independence case when $\delta\leq0.5$. For larger values of $\delta$, the range $\lambda$ is more variable and the smoothness parameter $\nu$ is slightly more biased, owing to the very strong dependence. However, in practice, one could restrict the parameter $\delta$ to $0\leq \delta\leq 0.8$, say, as $\delta>0.8$ is very unlikely to occur in applications. For all parameters $\delta$, $\lambda$ and $\nu$, but particularly for $\nu$, the fit improves significantly when more locations are available.

\subsubsection{Testing the dependence class}
A major advantage of model~\eqref{eq:model} over currently available models for spatial extremes is that we do not need to explicitly determine whether the data exhibit asymptotic dependence or asymptotic independence in order to select an appropriate class of models. However, since so much effort has previously been placed on determining the appropriate dependence class, we present the details and simulation experiments of a model-based test for this here. \citet{Coles.etal:1999} suggest using nonparametric estimators of the measure $\chi_u$ (defined slightly differently to \eqref{eq:chidef}) and its counterpart $\bar\chi_u$, but when the threshold $u$ increases to unity, the associated uncertainty inflates dramatically. This renders any test based on these nonparametric estimators almost useless in practice. To increase the power for discriminating between asymptotic dependence and independence, a parametric model-based approach seems sensible and our copula model \eqref{eq:model} provides a very natural way to proceed, because the transition between the two asymptotic paradigms takes place in the interior of the parameter space. We stress, however, that the validity of such a test is reliant on modeling assumptions, and as such is best used in conjunction with other diagnostics. Standard likelihood theory can be invoked to design tests for the null hypotheses
\begin{align*}
H_0^{\rm AD}&:\delta>1/2\mbox{ (asympt.\ dependence)}\quad &\mbox{\emph{vs}}&&\quad H_A^{\rm AI}&:\delta\leq 1/2\mbox{ (asympt.\ independence)};\\
H_0^{\rm AI}&:\delta\leq 1/2\mbox{ (asympt.\ independence)}\quad &\mbox{\emph{vs}}&&\quad H_A^{\rm AD}&:\delta>1/2\mbox{ (asympt.\ dependence)}.
\end{align*}
Let $\hat{\boldsymbol{\psi}}=(\hat\delta,\hat\lambda,\hat\nu)^T$ be the maximum likelihood estimator (MLE). We suggest using asymptotic normality of $\hat{\boldsymbol{\psi}}$ to test for $H_0^{\rm AD}$ or $H_0^{\rm AI}$, an assumption that should hold true if $n$ is large and $\delta$ is not too close to its boundaries $0$ and $1$. In particular, denoting the estimated variance of $\hat\delta$ by $\hat v_\delta$, the power of these tests at the level $100\times(1-\alpha)\%$ can be computed as 
\begin{align}
\Pr(\mbox{reject $H_0^{\rm AD}$}\mid \mbox{$H_A^{\rm AI}$ holds})&=\Pr(\hat\delta<0.5-\sqrt{\hat v_\delta}z_{1-\alpha}\mid \delta\leq 0.5),\label{eq:powerAD}\\
\Pr(\mbox{reject $H_0^{\rm AI}$}\mid \mbox{$H_A^{\rm AD}$ holds})&=\Pr(\hat\delta>0.5+\sqrt{\hat v_\delta}z_{1-\alpha}\mid \delta> 0.5),\label{eq:powerAI}
\end{align}
respectively, where $z_{1-\alpha}$ is the $(1-\alpha)$-quantile of the standard normal distribution. 

To compute the power curves \eqref{eq:powerAD} and \eqref{eq:powerAI}, we drew $300$ simulations of model \eqref{eq:model} at $d=2,5,10$ locations in $[0,1]^2$ with $n=1000,2000$ independent replicates, under the same setting as \S\ref{sec:ParameterEstimation}. Range $\lambda=0.5$ and smoothness $\nu=1$ were fixed, and we considered a sequence $\delta\in[0.3,0.8]$ in steps of $0.02$, estimating all parameters using the MLE based on \eqref{eq:LogLikelihood} with marginal thresholds $u_1^\star=\cdots=u_d^\star=0.95$. The Hessian matrix at the MLE was used in order to compute $\hat v_\delta$ as the $(1,1)$-entry of the reciprocal Fisher information.

\begin{figure}[t!]
\centering
\includegraphics[width=0.7\linewidth]{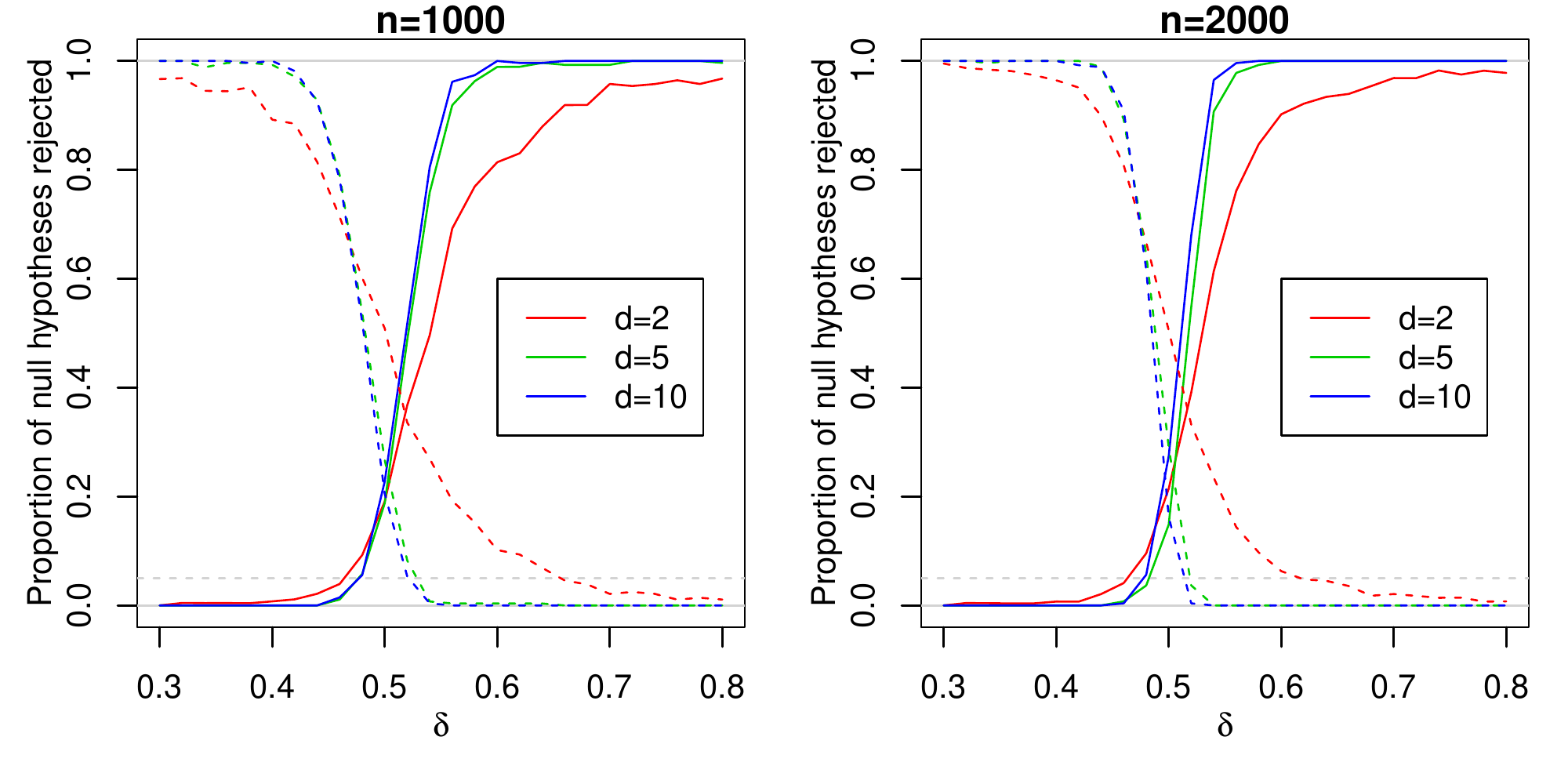}
\caption{Proportion of times the null hypotheses $H_0^{\rm AD}$ (dashed) and $H_0^{\rm AI}$ (solid) are rejected, plotted against the true $\delta$ value, computed from $300$ simulations. The horizontal dashed line at $0.05$ shows the nominal level used for the tests. The model~\eqref{eq:model} was simulated at $d=2$ (red), $5$ (green), and $10$ (blue) sites in $[0,1]^2$ with $n=1000$ (left), $2000$ (right) replicates, using an underlying Gaussian copula for $W(\bs)$ with correlation function $\rho(\bs_j,\bs_k)=\exp\{-(\|\bs_j-\bs_k\|/\lambda)^\nu\}$ and $\lambda=0.5$, $\nu=1$. Parameters were estimated by maximum likelihood based on \eqref{eq:LogLikelihood} with marginal thresholds $u_1^\star=\cdots=u_d^\star=0.95$.}\label{fig:powers}
\end{figure}
Figure~\ref{fig:powers} displays the proportion of null hypotheses rejected (i.e., the power curves \eqref{eq:powerAD} and \eqref{eq:powerAI} when the corresponding null hypotheses are false), estimated using the $300$ simulations and plotted as a function of $\delta\in[0.3,0.8]$. As expected, for all dimensions, the power to reject asymptotic dependence (respectively asymptotic independence) improves as $\delta\to0$ (respectively $\delta\to1$), and with higher dimensions, although there is little difference between $d=5$ and $d=10$. Comparing left and right panels, increased sample size also improves power, with a steeper transition around $\delta=1/2$. The departure from nominal levels for the Type I error however suggests that the Hessian may not give a good representation of the asymptotic variance, possibly owing to numerical approximations. In Section~\ref{sec:WaveHeight} we suggest using bootstrap methods to calculate uncertainty.

\section{Oceanographic applications}
\label{sec:Oceanographic}
\subsection{Hindcast significant wave height data}
\label{sec:WaveHeight}

\citet{WadsworthTawn12} considered modeling the extremes of the winter observations of a hindcast dataset of significant wave height, a measure of ocean energy, from the North Sea. Calculating the coefficient of tail dependence $\eta(h)$ for the wave height process, they suggested that there was evidence for asymptotic independence of the process, although strong spatial dependence between sites. Figure~\ref{fig:waveheightchi} suggests a high degree of ambiguity in what the appropriate extremal dependence structure should be, since the summary $\chi_u$ is decreasing as $u$ increases, but not necessarily to a value of zero. This ambiguous situation is replicated throughout numerous applications, and demonstrates the necessity for a model such as model~\eqref{eq:model} that can handle both scenarios.

Measurements of the hindcast are recorded at three-hourly intervals, yielding eight observations per day, over a period of $31$ years. In total the dataset of winter (December, January, February) wave heights consists of  $22376$ observations at $50$ locations. Margins are transformed to uniform using the semiparametric transformation of \citet{ColesTawn91}. The data are strongly temporally dependent and so we subsample to extract one realization per day, giving $2797$ observations. The resulting data still exhibit temporal dependence, but this thinning eases the computational burden of model fitting, whilst the information loss should be small. Finally, we select a subset of $20$ sites to fit the model to, whilst using all data for validation of the fit. Distance is measured in units of latitude (one unit $\approx 111$km); the range of distances between sites is $0.27$--$2.99$ units.

Model~\eqref{eq:model} was fitted by maximum likelihood based on \eqref{eq:LogLikelihood} with thresholds $u_1^\star=\cdots=u_d^\star=0.95$, assuming a Gaussian copula for the $W$ process (Example~\ref{ex:WGaussian}); Table~\ref{tab:Hs} reports the results. The uncertainty measures are based on 200 bootstrap samples, created using the stationary bootstrap \citep{PolitisRomano94}. This procedure relies on sampling blocks of geometric length; we sampled using an average length of 14 days, although any blocks that reached the end of February (i.e., the end of one winter) were curtailed, so that observations within a block are always consecutive. Figure~\ref{fig:Theta} in the Supplementary Material shows that this bootstrap procedure captures the temporal dependence in the extremes adequately.

The MLE of $\delta$ indicates asymptotic independence, although the 95\% bootstrap confidence interval includes values above 0.5, meaning that firm conclusions about the asymptotic dependence class are difficult to draw; this further highlights the need for models that can incorporate both scenarios. Whilst asymptotic independence is indicated, the value of $\delta$ suggests that that our model is more suited than a simple Gaussian model. To reinforce this, we also fit a Gaussian model, using the same censored likelihood scheme, with results reported on the right side of Table~\ref{tab:Hs}. Although the Gaussian model is nested within the model we fit, testing is non-standard as it occurs at the boundary of the parameter space, i.e., for $\delta=0$. The maximized log-likelihood for our model was 62 units higher than for the Gaussian model, representing a clear improvement, although interpretation is difficult as there is no explicit accounting for temporal dependence in the likelihood.

\begin{table}
\centering
\caption{MLEs, standard deviations (SD) of bootstrap replications, and approximate 95\% confidence intervals (CI) for parameters of our model~\eqref{eq:model} (left) and the Gaussian process (right). Figures given to the longer of two decimal places or significant figures.}
\label{tab:Hs}
 \begin{tabular}{llll|lll}\hline
  & MLE & SD & 95\% CI& MLE & SD & 95\% CI\\\hline
 $\delta$ & 0.46 & 0.039 & (0.36,0.54)  & -- & -- & --\\
 $\lambda$ & 3.19 & 0.26 & (2.60,3.71) & 3.84 & 0.17 & (3.62,4.26) \\ 
  $\nu$ & 1.98 & 0.0033 & (1.97,1.98)  & 1.97 & 0.0043 & (1.96,1.98)\\ \hline
\end{tabular}
\end{table}

To assess the fit of the model, we consider two diagnostics. Figure~\ref{fig:ChiHS20DFit} displays the fitted value of $\chi_u$, as defined in~\eqref{eq:chiu:intro}, for the subset of sites included in the model fit (left panel) and the subset of sites excluded from the fit (right panel). Although the model was fitted using censored likelihood above a $95\%$-quantile threshold, the fit looks good on the plotted range $u\in(0.9,1)$. The Gaussian model clearly underestimates the dependence.

\begin{figure}
 \centering
 \includegraphics[width=0.45\textwidth]{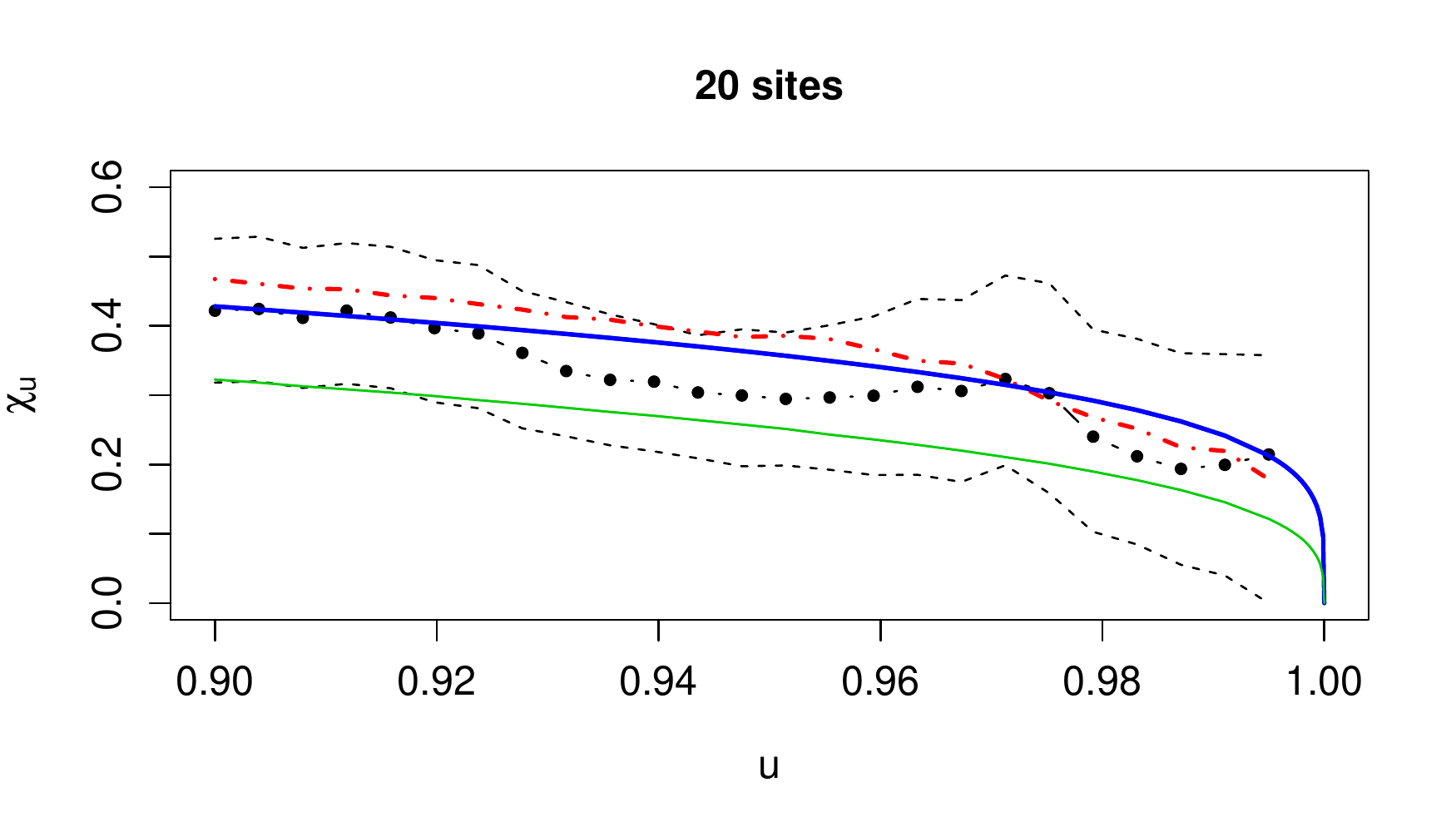}
 \includegraphics[width=0.45\textwidth]{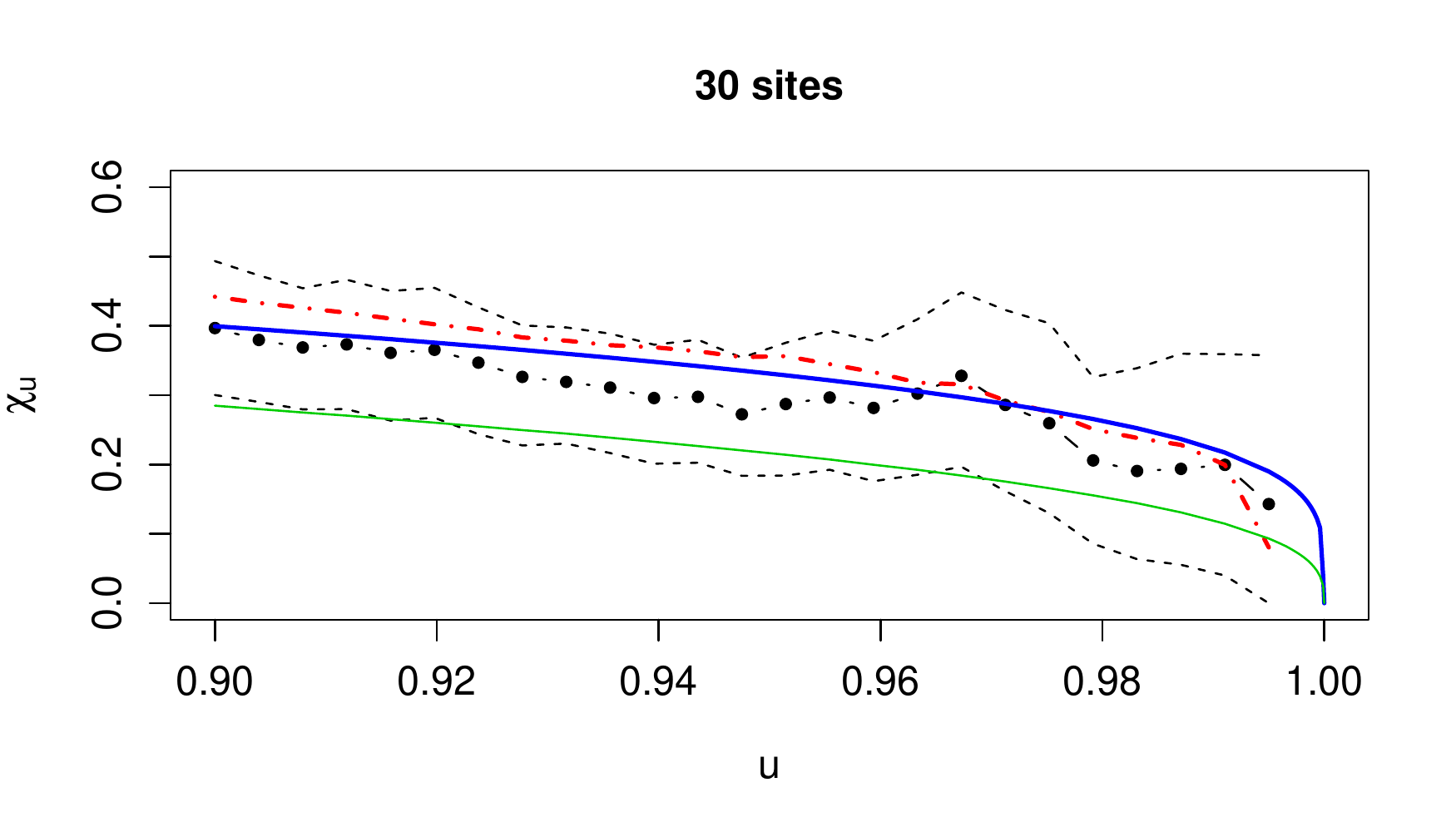}
 \caption{Estimates of $\chi_u$ for the hindcast wave height data. Central black dots: empirical estimate of $\chi_u$ from the temporally thinned data; dashed lines: approximate 95\% confidence intervals based on the stationary bootstrap procedure described in the text; dot-dash red line: empirical estimate of $\chi_u$ from all data; thick solid blue line: fit from our model; thin solid green line: fit from the Gaussian model. Left: data to which the model was fitted (from $20$ sites); right: data to which the model was not fitted (from $30$ sites).}\label{fig:ChiHS20DFit}
\end{figure}

The second diagnostic we consider is the distribution of the number of threshold exceedances, conditioning upon having at least one exceedance. The Supplementary Material contains histograms of the distribution from our data sample and from the fitted model, and suggests that the fitted model appears to capture this distribution quite well.

\subsection{Newlyn oceanographic data}
\label{sec:WaveSurge}
We fit a bivariate version of model~\eqref{eq:model}, as discussed in Example~\ref{sec:NonSpatial}, to the Newlyn oceanographic data analyzed in \citet{Wadsworthetal17} to illustrate an asymmetric construction, and to compare with the symmetric models fitted therein. The data, shown in Figure~\ref{fig:newlyn}, comprise $2894$ observations of wave height, surge and period, and we analyze them pairwise, transforming to uniformity again using the semiparametric transformation of \citet{ColesTawn91}. To generate an asymmetric model, we assume that the copula of $(W_1,W_2)^T$ is that of an inverted Dirichlet max-stable distribution (recall Example~\ref{ex:WinvertedMS}). The bivariate Dirichlet max-stable distribution \citep{ColesTawn91} has exponent function
\begin{align*}
 V(x_1,x_2) = \frac{1}{x_1}\left\{1-\mbox{Be}\left(\frac{\alpha x_1}{\alpha x_1 +\beta x_2};\alpha+1,\beta\right)\right\} + \frac{1}{x_2}\mbox{Be}\left(\frac{\alpha x_1}{\alpha x_1 + \beta x_2};\alpha,\beta+1\right),\qquad \alpha,\beta>0,
\end{align*}
where $\mbox{Be}(\cdot,a,b)$ is the Beta distribution function with shape parameters $a$ and $b$. The bivariate inverted max-stable distribution with Pareto margins has joint survivor function $\Pr(W_1>w_1,W_2>w_2) = \exp[-V\{(\log w_1)^{-1}, (\log w_2)^{-1}\}]$. To ensure consistency with the approach of \citet{Wadsworthetal17}, we use the censored likelihood described therein and at the end of \S\ref{sec:Likelihood} for both models. That is, we use the full density contribution when either variable is above a censoring threshold, which is set to the $95\%$-quantile in each margin. Table~\ref{tab:newlyn} gives the Akaike Information Criterion (AIC) for the model of~\citet{Wadsworthetal17} and our asymmetric model; improvements are seen for pairs involving wave period, which shows a more asymmetric dependence structure than height and surge. One limitation of this choice for $(W_1,W_2)^T$ is that it cannot exhibit negative dependence, and as such, the model is less flexible when it comes to accounting for dependence structures with weak asymptotic dependence (i.e., with small but positive $\chi$). This does not appear to be an issue for these asymptotically independent pairs, but alternative choices for $(W_1,W_2)^T$ such as the skew bivariate normal \citep{AzzaliniDallaValle96} could be used to overcome this.

\begin{table}
\centering
\caption{AIC for bivariate copula fits to the Newlyn data, using the symmetric model of \citet{Wadsworthetal17} (first row) and using our model with the asymmetric inverted Dirichlet model (second row).}
\label{tab:newlyn}
 \begin{tabular}{llll}\hline
  & Height--Surge & Period--Surge & Height--Period\\\hline
 AIC WTDE & 264.1 & 515.4 & 225.7\\
 AIC asymmetric & 267.3 & 493.8 & 181.6\\ \hline
\end{tabular}
\end{table}

\begin{figure}
 \centering
 \includegraphics[width=0.3\textwidth]{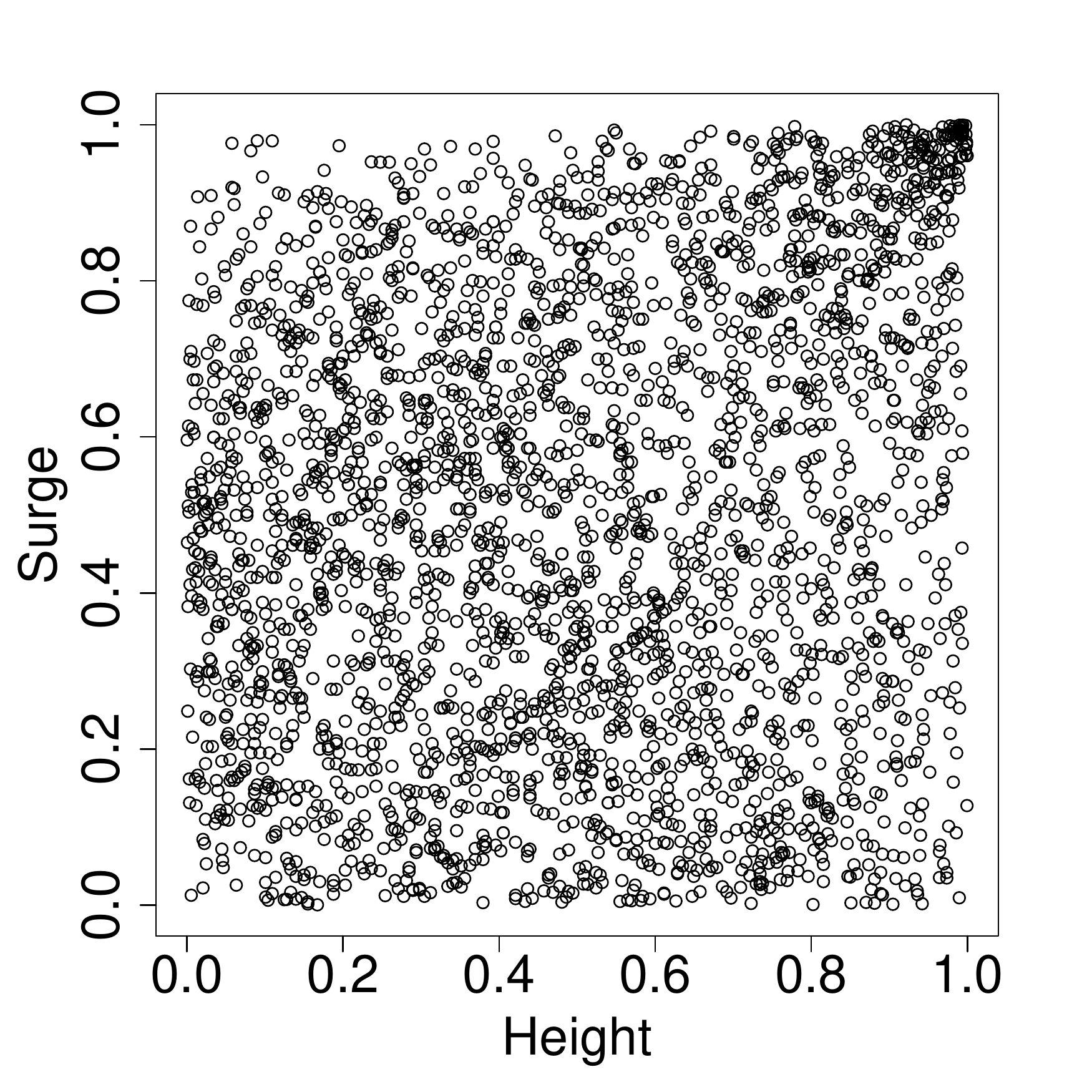}
 \includegraphics[width=0.3\textwidth]{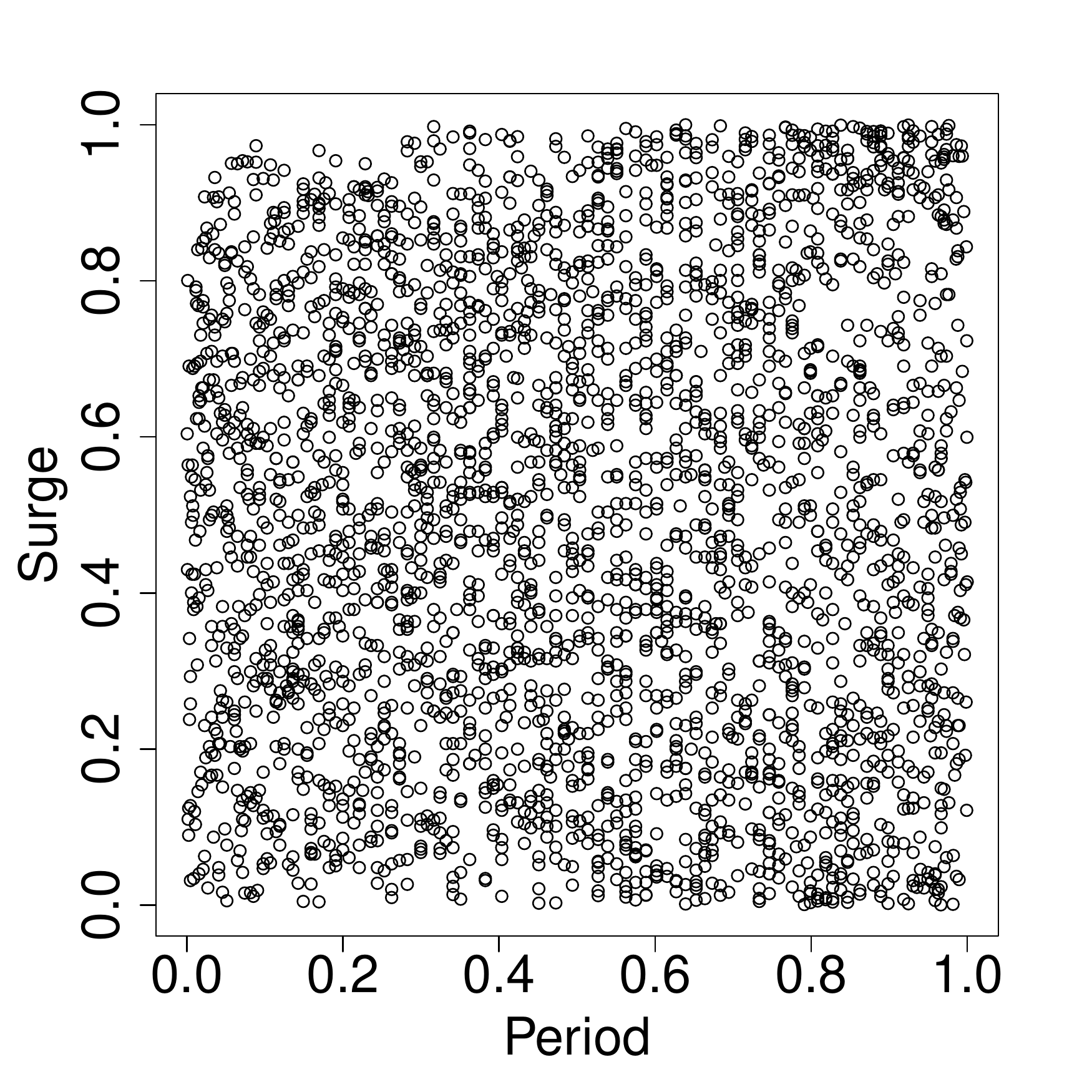}
 \includegraphics[width=0.3\textwidth]{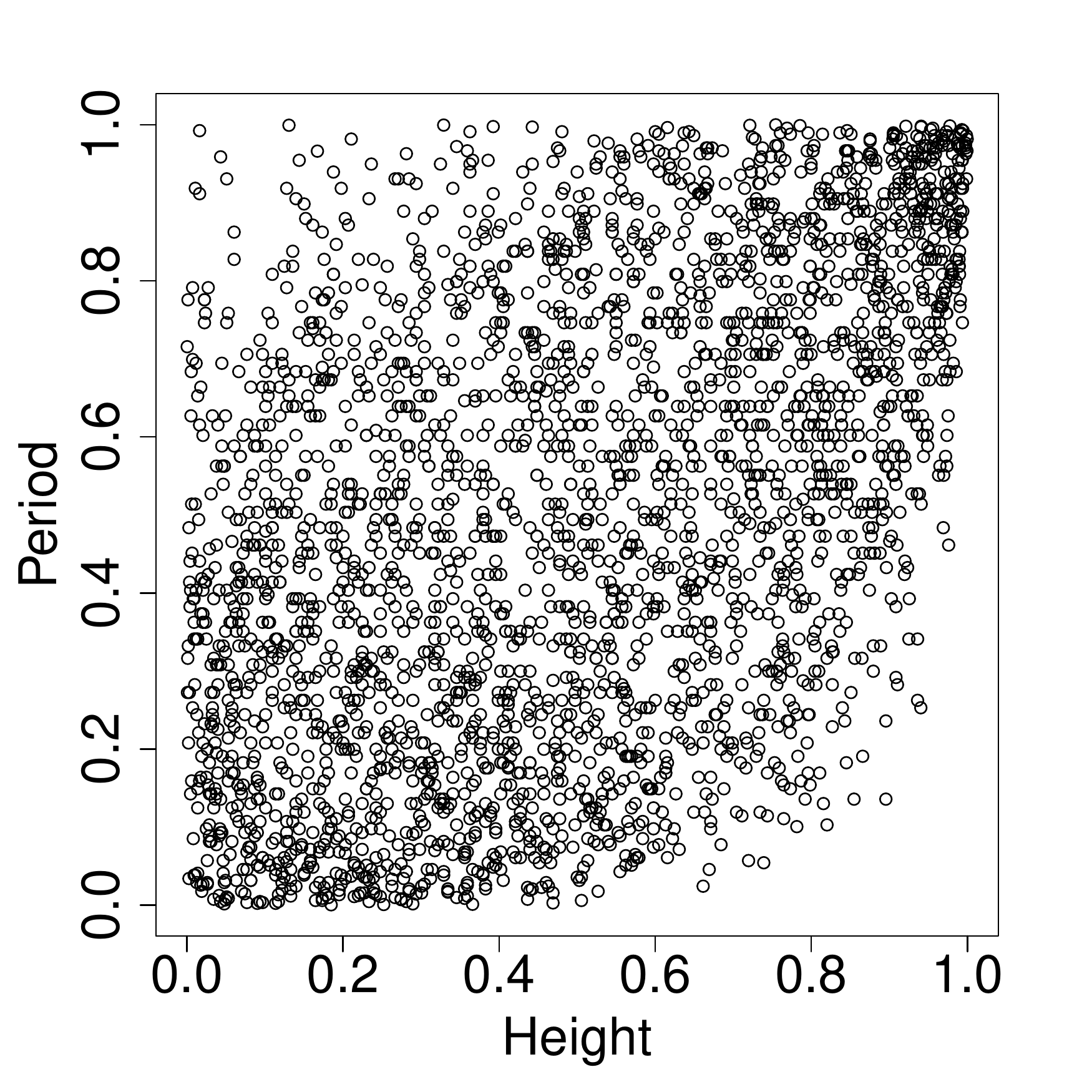}
 \caption{Newlyn wave data on approximate standard uniform margins. From left to right: Height-Surge, Period-Surge, Height-Period.} \label{fig:newlyn}
\end{figure}

\section{Discussion}
\label{sec:Discussion}
Motivated by deficiencies in existing frameworks for modeling spatial extremes, we presented a parsimonious model that is able to capture the sub-asymptotic dependence behavior of spatial processes. Importantly, both extremal dependence classes are captured, with rich structures within each class, and a smooth transition between paradigms at the interior of the parameter space. 

Inference for model~\eqref{eq:model} is feasible in moderate dimensions, but computationally intensive when $W$ has a Gaussian copula, owing to the need to integrate expressions involving a multivariate Gaussian distribution function. However, new quasi-Monte Carlo algorithms, such as those used by \citet{deFondevilleDavison16}, and the associated R package \texttt{mvPot}, have the potential to increase scalability; their code was used to speed up the bootstrap procedure in \S~\ref{sec:WaveHeight}. With the exception of the specific model used in \citet{deFondevilleDavison16}, truly high-dimensional inference for spatial extreme-value models has yet to be achieved, and our model is competitive with others in this aspect.

There are two notable limitations of the model~\eqref{eq:model}. The first of these is that for $\delta > 1/3$, $\eta_X(h)>1/2$ indicating a persistence of positive extremal association even as the lag $h\to\infty$. This is, however, a common problem with many models for spatial extremes. Consequently, the model is more suitable for smaller spatial regions or data for which this is not an issue. The second limitation concerns the link between $\delta$ and the limiting value of $\chi(h)$ for $\delta>1/2$. Since $W(\bs)\geq1$ we have $\min(W_j,W_k)\geq 1$ and consequently from~\eqref{eq:chiW}, $\chi(h) \geq (2\delta-1)/\delta$. As can be observed from Figure~\ref{fig:chi.eta} and equation~\eqref{eq:chiconv}, for values of $\delta$ near 1, the process~\eqref{eq:model} behaves similarly to a generalized Pareto process. However, model~\eqref{eq:model} would be unable to capture a weakly dependent generalized Pareto process, i.e., one for which $\chi_u(h)$ is constant in $u$ but its limit $\chi(h)$ is small and positive. In practice however, this is not likely to be restrictive, since in our experience almost all environmental datasets display a decreasing $\chi_u$ function.

 \subsection*{Acknowledgements} We thank Philip Jonathan of Shell Research for the wave height data analyzed in \S\ref{sec:WaveHeight}, and Rapha\"{e}l de Fondeville for helpful discussions and code for multivariate Gaussian computation. J. Wadsworth gratefully acknowledges funding from EPSRC fellowship grant EP/P002838/1.

\paragraph{Code and data} Code for fitting the models described is available as Supplementary Material and at {\tt http://www.lancaster.ac.uk/$\sim$wadswojl/SpatialADAI}. The NEXTRA hindcast data analyzed in \S~\ref{sec:WaveHeight} are subject to restrictions. Access may be granted for academic purposes by members of the North European Storm Study User Group (NUG); requests can be made using the details at {\tt http://www.oceanweather.com/metocean/next/index.html}. The Newlyn wave data analyzed in \S~\ref{sec:WaveSurge} are available as Supplementary Material.

\appendix
\section{Proofs}
\label{app:proofs}
\begin{proof}[Proof of Proposition~\ref{prop:joint}]
\begin{align*}
 \Pr(X_j>x, X_k>x) &= \Pr(W_j > x^{1/(1-\delta)}R^{-\delta/(1-\delta)}, W_k > x^{1/(1-\delta)}R^{-\delta/(1-\delta)})\\
 &=\Pr(W_j > S, W_k > S),
\end{align*}
where $S= x^{1/(1-\delta)}R^{-\delta/(1-\delta)}$, so that $S$ has support $(0,x^{1/(1-\delta)})$, and Lebesgue density $f_S(s) = \frac{1-\delta}{\delta}s^{(1-\delta)/\delta-1}x^{-1/\delta}$ on this interval. Using assumption~\eqref{eq:Whrv}, we have
\begin{align*}
 \Pr(W_j > S, W_k > S) &= \frac{1-\delta}{\delta} x^{-1/\delta} \int_0^1 s^{(1-\delta)/\delta-1}\,{\rm ds} +\frac{1-\delta}{\delta} x^{-1/\delta}  \int_{1}^{x^{1/(1-\delta)}} L_W(s) s^{(1-\delta)/\delta-1/\etaW-1} \,{\rm ds}\\
 &=x^{-1/\delta} + \frac{1-\delta}{\delta} x^{-1/\delta} \int_{1}^{x^{1/(1-\delta)}} L_W(s) s^{(1-\delta)/\delta-1/\etaW-1} \,{\rm ds}.
\end{align*}

Consider the behavior of $\int_{1}^{x^{1/(1-\delta)}} L_W(s) s^{(1-\delta)/\delta-1/\etaW-1} \,{\rm ds}$, which is convergent since we have a well defined probability. We will apply Karamata's Theorem \citep[][Theorem 2.1]{Resnick06} and so distinguish between the cases when the index of regular variation is $\gtreqless -1$. The notation $g\in RV_{\rho}$ denotes that a function $g$ is regularly varying at infinity with index $\rho\in\mathbb{R}$.

\paragraph*{Case 1: $(1-\delta)/\delta-1/\etaW-1 \geq -1$} i.e., $\etaW \geq \delta/(1-\delta)$. By Karamata's Theorem $\int_1^{x} L(s)s^{\theta}\, ds \in RV_{\theta+1}$ when $\theta \geq -1$. Thus 
\begin{align*}
 \int_{1}^{x^{1/(1-\delta)}} L_W(s) s^{(1-\delta)/\delta-1/\etaW-1} \,{\rm ds} = \tilde{L}(x)x^{1/\delta-1/\{\etaW(1-\delta)\}},
\end{align*}
where $\tilde{L}$ is a new SV function, using also a result on composition of regularly varying functions \citep[][Prop.\ 2.6 (iv)]{Resnick06}.
\medskip\noindent
Overall in Case~1 we thus have
\begin{align*}
  \Pr(X_j>x, X_k>x) &= L(x) x^{-1/\{\etaW(1-\delta)\}},
\end{align*}
for some slowly varying function $L$, noting that terms of order $x^{-1/\delta}$ are absorbed into $L$ when $\etaW>\delta/(1-\delta)$.


\paragraph*{Case 2: $(1-\delta)/\delta-1/\etaW-1 < -1$} i.e., $\etaW < \delta/(1-\delta)$. By Karamata's Theorem $\int_{x}^{\infty} L(s)s^{\theta}\, ds \in RV_{\theta+1}$ when $\theta < -1$. We have
\begin{align}
 \int_{1}^{x^{1/(1-\delta)}} L_W(s) s^{(1-\delta)/\delta-1/\etaW-1} \,{\rm ds} = \int_{1}^{\infty} L_W(s) s^{(1-\delta)/\delta-1/\etaW-1} \,{\rm ds} - \int_{x^{1/(1-\delta)}}^{\infty} L_W(s) s^{(1-\delta)/\delta-1/\etaW-1} \,{\rm ds}, \label{eq:case2expansion}
\end{align}
and so the second term on the right-hand side the is regularly varying of index $1/\delta-1/\{\etaW(1-\delta)\}$.
\medskip
The first term on the right-hand side of expression~\eqref{eq:case2expansion} is established by noting that
\begin{align*}
 \E\{\min(W_j,W_k)^{(1-\delta)/\delta}\} &= \int_0^{\infty} \Pr\{\min(W_j,W_k)^{(1-\delta)/\delta} > t\} \,{\rm dt}\\
 &=1+\int_1^{\infty} L_W\{t^{\delta/(1-\delta)}\}t^{-\delta/\{\etaW(1-\delta)\}} \,{\rm dt} = 1+\frac{1-\delta}{\delta}\int_{1}^{\infty} L_W(s) s^{(1-\delta)/\delta-1/\etaW-1} \,{\rm ds}.
\end{align*}
\medskip
Overall in Case~2 we thus have
\begin{align*}
  \Pr(X_j>x, X_k>x) &= \E\{\min(W_j,W_k)^{(1-\delta)/\delta}\}x^{-1/\delta} - L(x)x^{-1/\{\etaW(1-\delta)\}}\\
  &=\E\{\min(W_j,W_k)^{(1-\delta)/\delta}\}x^{-1/\delta}\{1+o(1)\},
\end{align*}
since $\etaW(1-\delta)<\delta$.
\end{proof}

\begin{proof}[Proof of Corollary~\ref{cor:joint}]
Since $X$ has common margins and upper endpoint infinity, the extremal dependence class is determined by the limit
\begin{align*}
\chi= \lim_{x\to \infty} \frac{\Pr(X_j>x, X_k>x)}{\Pr(X_j>x)}.
\end{align*}

\paragraph{If $\delta>1/2$:} Then $\Pr(X_j>x)\sim \frac{\delta}{2\delta-1} x^{-1/\delta}$. We have $\delta/(1-\delta)>1$ so we must be in Case~2, and 
\begin{align*}
 \chi =  \E\{\min(W_j,W_k)^{(1-\delta)/\delta}\} \frac{2\delta-1}{\delta} >0,
\end{align*}
with expression~\eqref{eq:chiW} following as 
\begin{align}
 \E\{W_j^{(1-\delta)/\delta}\} = \int_{1}^{\infty} w^{(1-\delta)/\delta-2} \,{\rm d}w = \frac{\delta}{2\delta-1}. \label{eq:EWa}
\end{align}

\paragraph{If $\delta=1/2$:}
Consider the relation~\eqref{eq:etaX}, and note that since $\Pr(X_j>x)^{-1}$ is regularly varying with limit infinity, then the composition $L_X\{\Pr(X_j>x)^{-1}\}=:L^*(x)$ is slowly varying at infinity; cf.\ \citet[][Prop.\ 2.6(iv)]{Resnick06}. For $\delta=1/2$, we have
\[
 \Pr(X_j>x) = x^{-2}\{2\log(x) + 1\},\quad \Pr(X_j>x,X_k>x) = \E\{\min(W_j,W_k)\}x^{-2}\{1+o(1)\},
\]
since $\etaW<1$ by assumption, which puts us in Case 2. We thus have 
\[
 \Pr(X_j>x,X_k>x) = L^*(x)\Pr(X_j>x),
\]
so that $\etaX=1$ and
\[
 L^*(x) \sim \frac{\E\{\min(W_j,W_k)\}}{\{2\log(x) + 1\}} \to 0,~~x\to\infty.
\]

\paragraph{If $\delta<1/2$} Then $\Pr(X_j>x)\sim \frac{1-\delta}{1-2\delta} x^{-1/(1-\delta)}$. If $\etaW \geq \delta/(1-\delta)$ then we are in Case~1 and the survivor function is $L(x)x^{-1/\{\etaW(1-\delta)\}}$. Otherwise if $\etaW < \delta/(1-\delta)$ we are in Case~2 and the survivor function decays like $x^{-1/\delta}$. In both cases this leads to $\chi=0$ with coefficient of tail dependence, 
\begin{align*}
 \etaX &=\begin{cases}
           \delta/(1-\delta)& \mbox{if } \etaW < \delta/(1-\delta) \\
           \etaW & \mbox{otherwise}.
          \end{cases}
\end{align*}

\end{proof}

During the proofs of Propositions~\ref{prop:chiV} and~\ref{prop:chiconv}, we will need results on the quantile function $q(t):=F_{X}^{-1}\{1-1/t\}$, which we give in the following Lemma.
\begin{lem}
\label{lem:q}
 For $\delta>1/2$, the marginal quantile function $q(t)=F_{X}^{-1}(1-1/t)$ satisfies
 \begin{align*}
  q(t) = \left(\frac{\delta}{2\delta-1}\right)^{\delta}t^{\delta} \left[1-(1-\delta)\left(\frac{\delta}{2\delta-1}\right)^{(1-2\delta)/(1-\delta)}t^{(1-2\delta)/(1-\delta)}\{1+o(1)\}\right],~~~t\to\infty. 
 \end{align*}

\end{lem}

\begin{proof}
 
The quantile function is obtained by solving $1-F_{X}\{q(t)\} = t^{-1}$, which leads to
\begin{align*}
 \frac{\delta}{2\delta-1}q(t)^{-1/\delta}\left\{1-\frac{1-\delta}{\delta}q(t)^{(1-2\delta)/\{\delta(1-\delta)\}}\right\} &= t^{-1}
 \end{align*}
 and thus
 \begin{align}
 q(t)\left\{1-\frac{1-\delta}{\delta}q(t)^{(1-2\delta)/\{\delta(1-\delta)\}}\right\}^{-\delta} &= \left( \frac{\delta}{2\delta-1}\right)^{\delta}t^{\delta}. \label{eq:q}
 \end{align}
 Since $q(t)\to\infty$ as $t\to\infty$, and $(1-2\delta)/\{\delta(1-\delta)\}<0$ for $\delta>1/2$, expression~\eqref{eq:q} leads to
\begin{align*} 
 q(t) &=  \left(\frac{\delta}{2\delta-1}\right)^{\delta}t^{\delta}\{1+o(1)\},\quad t\to\infty,
\end{align*}
which can be fed back into~\eqref{eq:q} to give the result claimed.
\end{proof}

\begin{proof}[Proof of Proposition~\ref{prop:chiV}]
 
When $\delta>1/2$, the exponent function $V(x_1,\ldots,x_d)$ is obtained from the limit
\begin{align*}
 V(x_1,\ldots,x_d) = \lim_{t\to\infty} t(1-\Pr[X_1 \leq F_{X}^{-1}\{1-(tx_1)^{-1}\},\ldots,X_1 \leq F_{X}^{-1}\{1-(tx_d)^{-1}\}]).
\end{align*}
Using Lemma~\ref{lem:q}, we have $q(tx)= \left(\frac{\delta}{2\delta-1}\right)^{\delta}(tx)^{\delta}\{1+o(1)\}$, and so
\begin{align*}
1-\Pr\{X_1 \leq q(tx_1),\ldots,X_1 \leq q(tx_d)\} &= \Pr\left\{\max_{j=1,\ldots,d} \frac{X_j}{q(tx_j)} > 1\right\}\\
 &=\Pr\left[\max_{j=1,\ldots,d} \frac{R^{\delta}W_j^{1-\delta}}{\left(\frac{\delta}{2\delta-1}\right)^{\delta}(tx_j)^{\delta}\{1+o(1)\}} > 1\right]\\
 &= \int_{0}^{1} \Pr\left[\max_{j=1,\ldots,d} \frac{W_j^{(1-\delta)/\delta}}{\left(\frac{\delta}{2\delta-1}\right)x_j \{1+o(1)\}} > t u\right]\,{\rm du}\\
 &=\frac{1}{t}\int_{0}^{t} \Pr\left[\max_{j=1,\ldots,d} \frac{W_j^{(1-\delta)/\delta}}{\left(\frac{\delta}{2\delta-1}\right)x_j \{1+o(1)\}} > z\right]\,{\rm dz}.
\end{align*}
For sufficiently large $t$, an integrable function of the form $\Pr\left[K\max_{j=1,\ldots,d} \frac{W_j^{(1-\delta)/\delta}}{\left(\frac{\delta}{2\delta-1}\right)x_j} > z\right]$, $1<K<\infty$, dominates the integrand over $(0,\infty)$ and thus the above integral tends to 
\[
\int_{0}^{\infty} \Pr\left\{\max_{j=1,\ldots,d} \frac{W_j^{(1-\delta)/\delta}}{\left(\frac{\delta}{2\delta-1}\right)x_j} > z\right\}\,{\rm dz} = \E\left\{\max_{j=1,\ldots,d} \frac{W_j^{(1-\delta)/\delta}}{x_j}\right\}\left(\frac{2\delta-1}{\delta}\right),
\]
and hence
\begin{align*}
\lim_{t\to\infty} t[1-\Pr\{X_1 \leq q(tx_1),\ldots,X_d\leq q(tx_d)\}] &= \E\left\{\max_{j=1,\ldots,d} \frac{W_j^{(1-\delta)/\delta}}{x_j}\right\}\left(\frac{2\delta-1}{\delta}\right)\\
&=\E\left[\max_{j=1,\ldots,d} \frac{W_j^{(1-\delta)/\delta}}{\E\{W_j^{(1-\delta)/\delta}\}x_j}\right],
\end{align*}
the final line following by equation~\eqref{eq:EWa}.
\end{proof}

\begin{proof}[Proof of Proposition~\ref{prop:chiconv}]
The function 
\begin{align*}
\chi_u = \Pr\{F_X(X_j)>u \mid F_X(X_k)>u\} = \frac{\Pr\{X_j>F_X^{-1}(u),X_k>F_X^{-1}(u)\}}{1-u}.
\end{align*}
Lemma~\ref{lem:q} gives the behavior of $F_X^{-1}(u) = q\{(1-u)^{-1}\}$, whilst the proof of Proposition~\ref{prop:joint} provides $\Pr(X_j>x, X_k>x) = \E\{\min(W_j,W_k)^{(1-\delta)/\delta}\}x^{-1/\delta} - L(x)x^{-1/\{\etaW(1-\delta)\}}$, giving
\begin{align*}
 \frac{\Pr\{X_j>F_X^{-1}(u),X_k>F_X^{-1}(u)\}}{1-u} &= \E\{\min(W_j,W_k)^{(1-\delta)/\delta}\}\frac{F_X^{-1}(u)^{-1/\delta}}{1-u} - L\{F_X^{-1}(u)\}\frac{F_X^{-1}(u)^{-1/\{\etaW(1-\delta)\}}}{1-u}\\
 &= \chi \left[1+\frac{1-\delta}{\delta}\left(\frac{\delta}{2\delta-1}\right)^{(1-2\delta)/(1-\delta)}(1-u)^{(2\delta-1)/(1-\delta)}\{1+o(1)\}\right] \\
 & \qquad -L\{(1-u)^{-1}\}(1-u)^{\delta/\{\etaW(1-\delta)\}-1}\{1+o(1)\}.
\end{align*}
with constant terms absorbed in to $L$. Since $(2\delta-1)/(1-\delta)<\delta/\{\etaW(1-\delta)\}-1$ for $\etaW<1$, the result follows.

\end{proof}

\newpage
\section{Supplementary Material}
\subsection{Supporting information for Section~\ref{sec:InferenceSimulation}}

\begin{figure}[h]
 \centering
\includegraphics[width=0.8\linewidth]{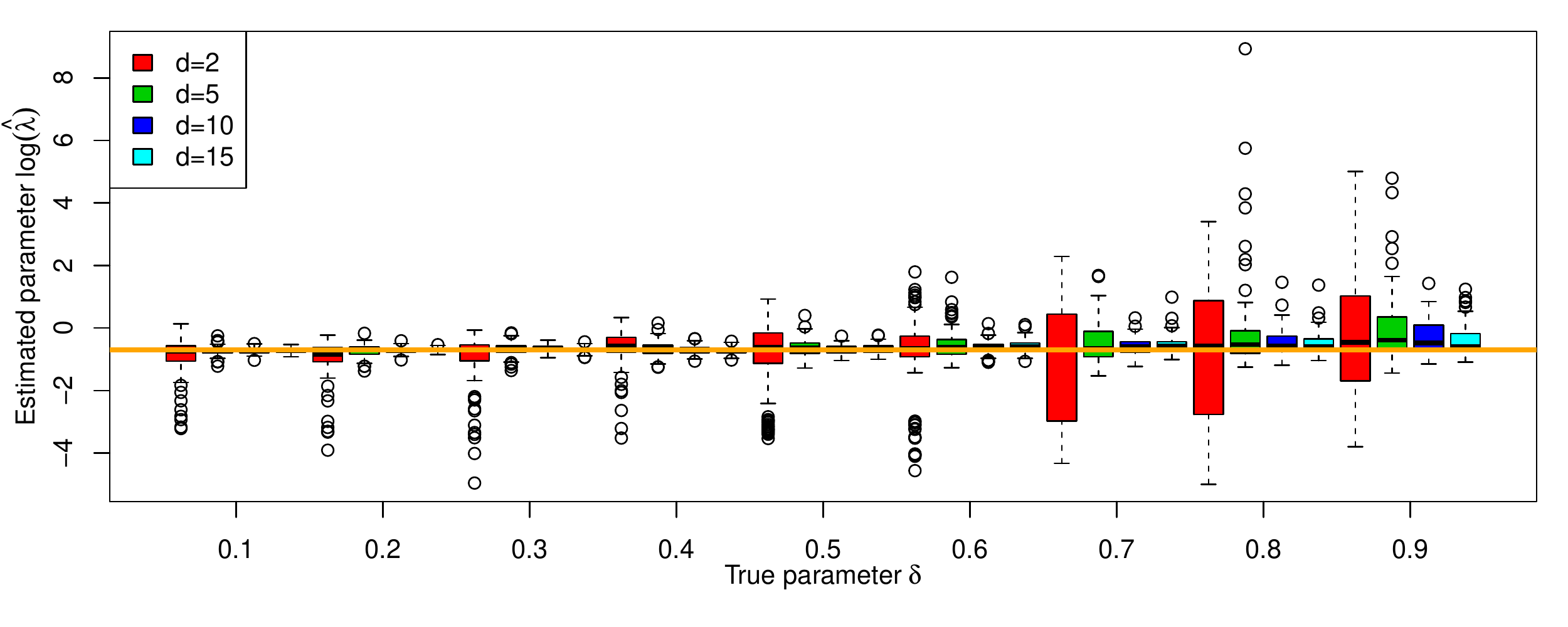}\\
\includegraphics[width=0.8\linewidth]{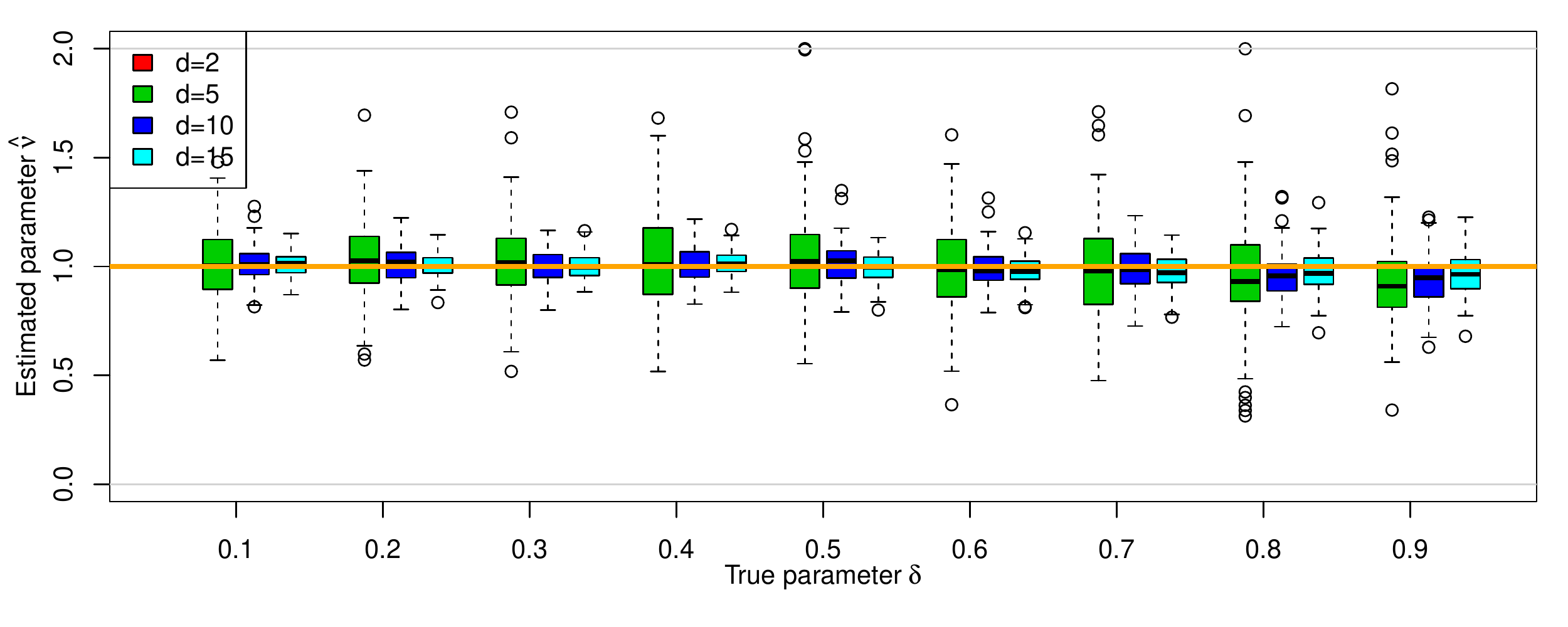}
\caption{Boxplots for the MLEs $\log(\hat{\lambda})$ and $\hat{\nu}$, estimated concurrently with $\hat{\delta}$ as in Figure~\ref{fig:results1} of Section~\ref{sec:Simulation}.}
\end{figure}

\pagebreak
\subsection{Supporting information for Section~\ref{sec:Oceanographic}}
\subsubsection{Bootstrap procedure}
To demonstrate that the stationary bootstrap procedure described in \S\ref{sec:WaveHeight} adequately reproduces the temporal dependence in the extremes, we consider a spatial extension of the \emph{extremal index} for univariate time series. For a stationary time series $\{X_t\}$, the extremal index, $\theta \in[0,1]$, can be defined as
\begin{align*}
\theta = \lim_{n\to\infty} \Pr(X_2\leq u_n,\ldots,X_{p_n}\leq u_n | X_1>u_n),
\end{align*}
where $p_n=o(n)$ and $u_n$ is a series such that $n\{1-F(u_n)\}\to\tau \in(0,\infty)$. The extremal index describes the degree of temporal clustering of extremes, with $1/\theta$ the limiting mean cluster size. A popular estimator for $\theta$ is the so-called Runs Estimator \citep{SmithWeissman94}. The estimate is formed by taking the reciprocal of the mean cluster size, whereby threshold exceedances are determined to be part of different clusters (the same cluster) if they are separated by a run of at least $m$ (fewer than $m$) consecutive non-exceedances. 

In our application we have a time series of spatial processes $\{X_t(s)\}$, which, as we consider winter months only, may reasonably be deemed stationary. In analogy to the univariate case, we define clusters of spatial threshold exceedances as follows. A realization of the process is deemed to be a ``threshold exceedance'' if the observation at any site exceeds a given threshold. Clusters are then defined as sequences of threshold exceedances separated by a run of at least $m$ non-exceedances, and $\theta$ as the reciprocal mean cluster size. Figure~\ref{fig:Theta} displays a histogram of estimated $\theta$s, using a value of $m=1$, from 200 bootstrap samples, along with that from the original dataset of 50 sites temporally thinned to one observation per day. The threshold value used was the 95\%-quantile, as in the model fit. The agreement between the original and bootstrap samples indicates that the temporal structure of the extremes is adequately reproduced.
\begin{figure}[h]
 \centering
\includegraphics[width=0.6\linewidth]{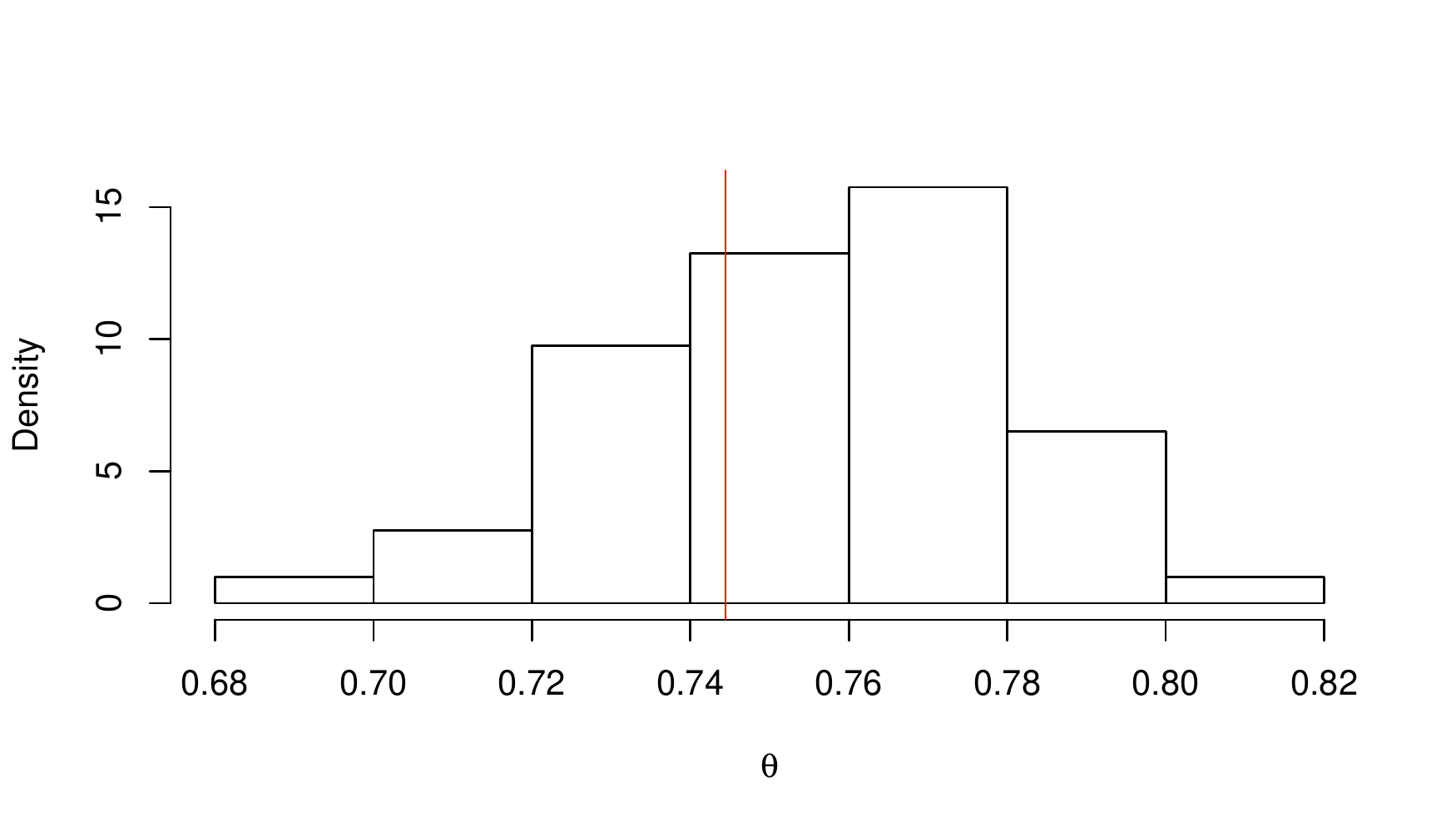}
\caption{Estimates of the extremal index from the time series of spatial processes using the stationary bootstrap sampling procedure described in \S\ref{sec:WaveHeight}. The vertical line is the value from the original sample.}
\label{fig:Theta}
\end{figure}

\subsubsection{Additional model fit diagnostics}
\begin{figure}[h]
 \centering
\includegraphics[width=0.45\linewidth]{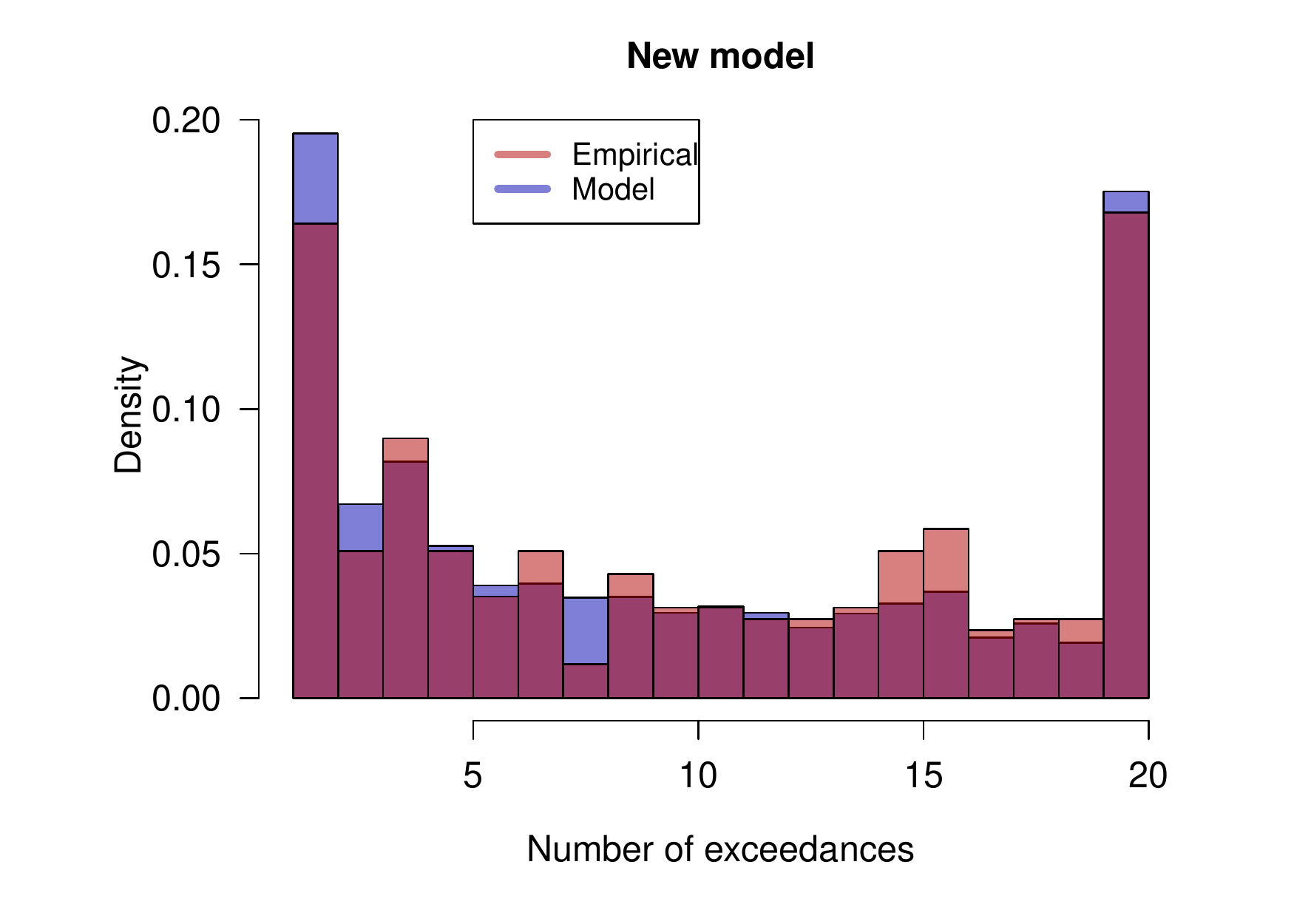}
\includegraphics[width=0.45\linewidth]{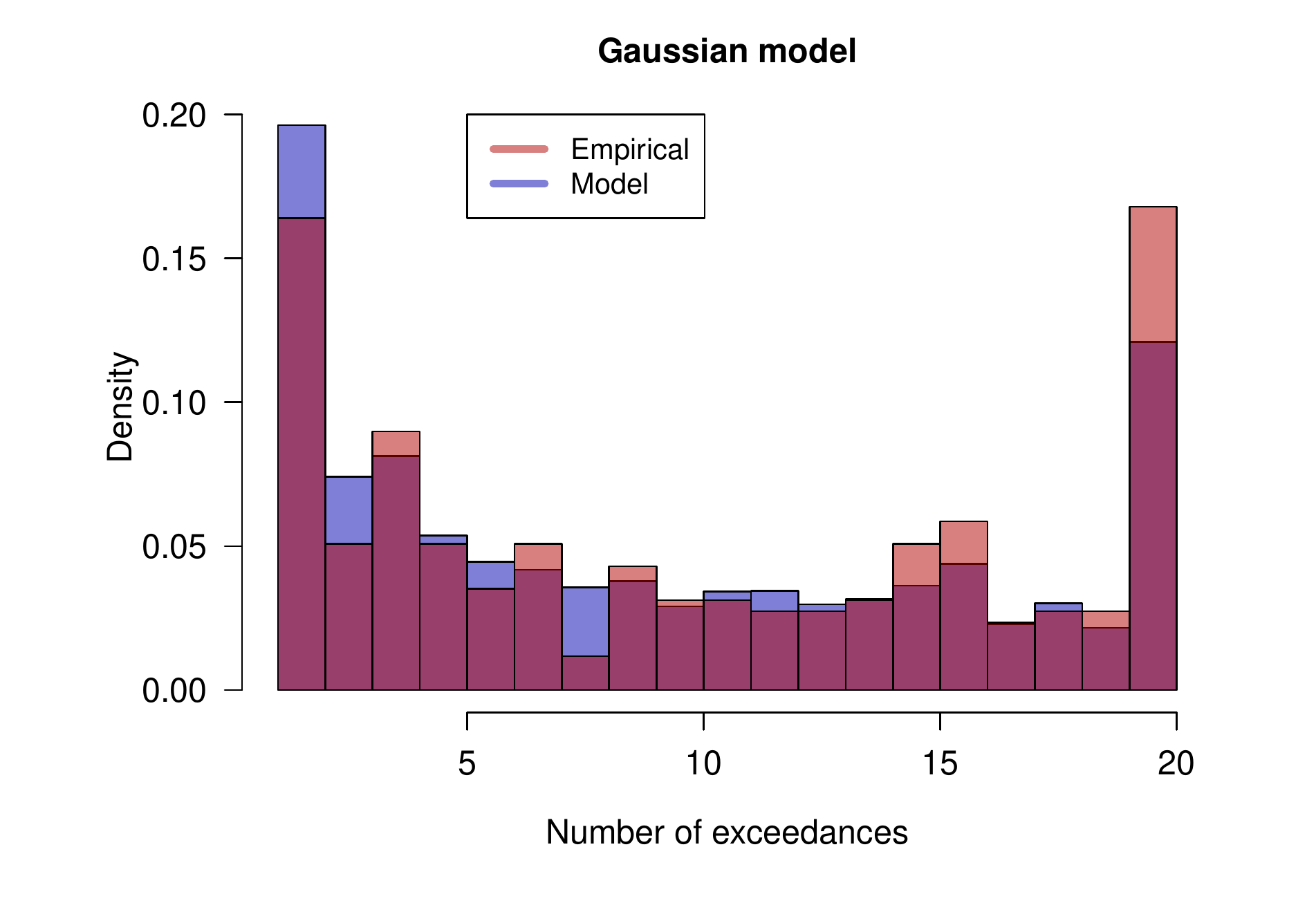}\\
\includegraphics[width=0.45\linewidth]{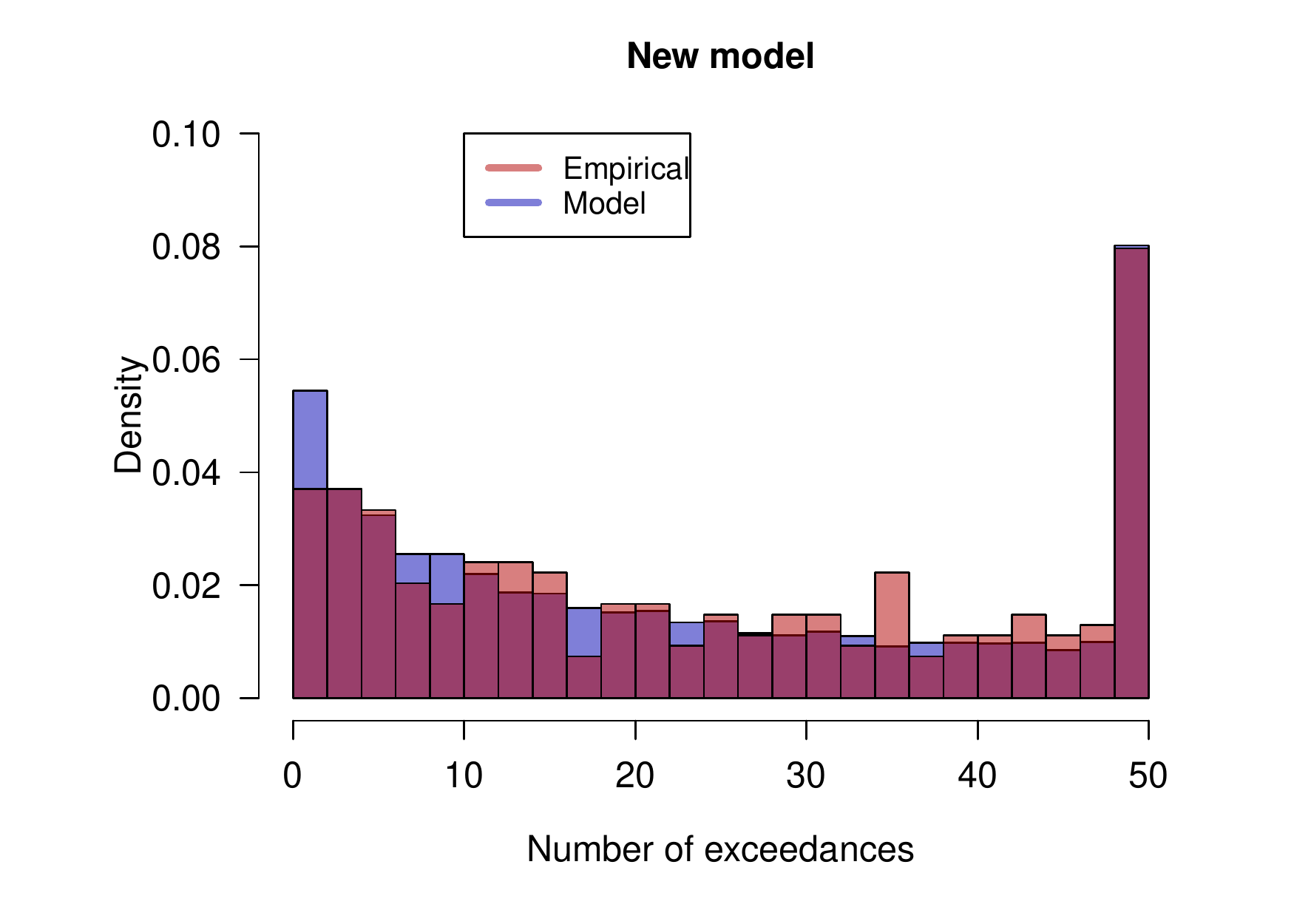}
\includegraphics[width=0.45\linewidth]{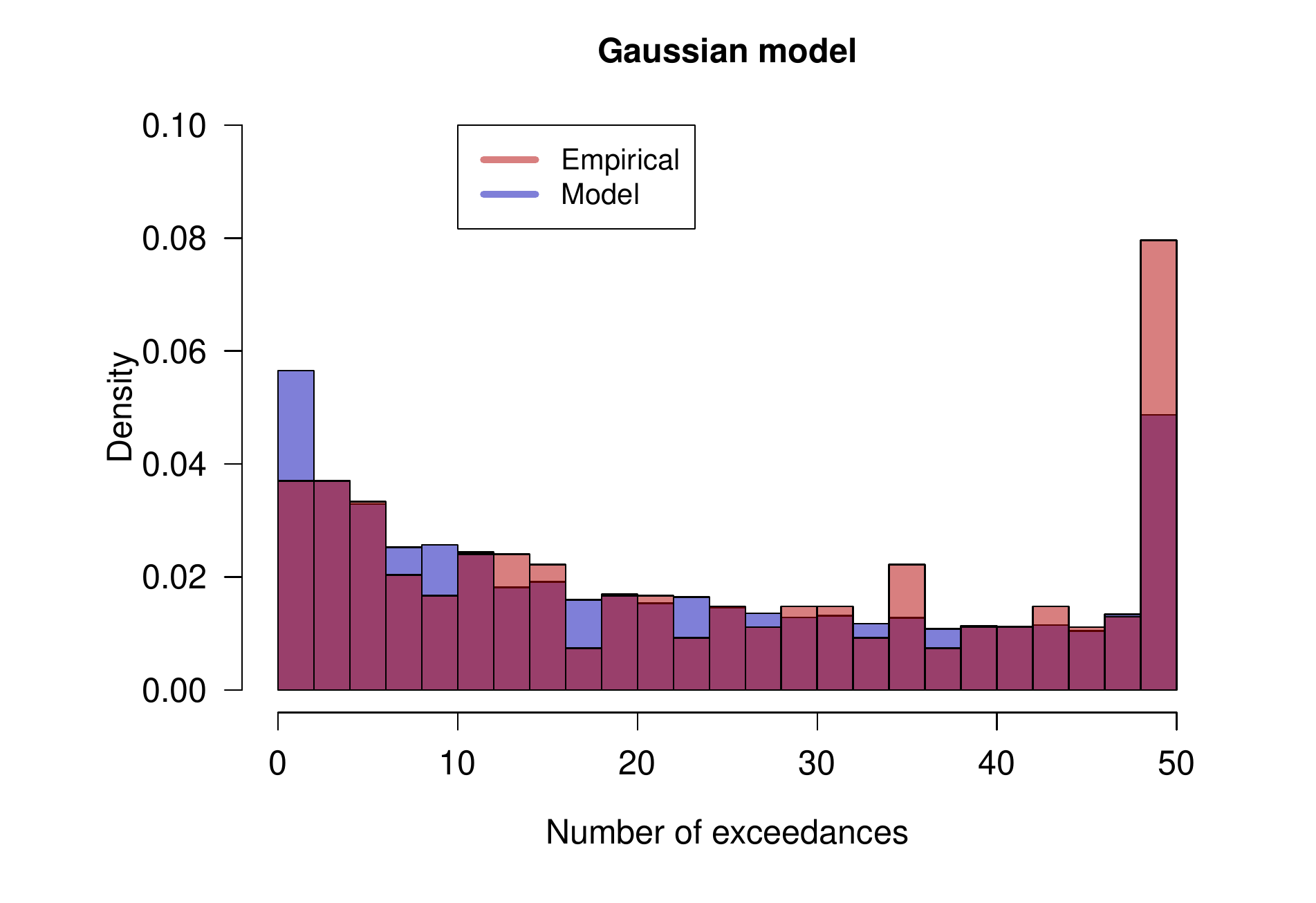}
\caption{Distribution of the number of exceedances, given at least one exceedance of the $95\%$-quantile threshold. These histograms are based on the data at the $20$ sites used for fitting the model (top) or all $50$ sites (bottom). Model-based quantities are calculated for our new model (left) and the Gaussian model (right) by simulating $10^5$ values from the fitted dependence models. }\label{fig:Hist}
\end{figure}

\baselineskip 20pt

\bibliographystyle{apalike}
\bibliography{SpatialADAIBib.bib}

\baselineskip 10pt

\end{document}